\documentclass[a4paper,11pt,onecolumn,oneside]{article}
\usepackage{soul}
\usepackage[framemethod=tikz]{mdframed}
\usepackage{amsfonts}
\usepackage[british]{babel}
\usepackage{amsmath}
\usepackage{booktabs}
\usepackage{mathrsfs}
\usepackage{amssymb}
\usepackage{mathrsfs}
\usepackage{amssymb}
\usepackage{yhmath}
\usepackage{latexsym}
\usepackage{makeidx}
\usepackage{bm}
\usepackage{cases}
\usepackage[round]{natbib}
\usepackage{cite}
\usepackage{tgtermes}
\usepackage[T1]{fontenc}
\usepackage[]{fncychap}
\usepackage[latin9]{inputenc}
\usepackage{textcomp}
\usepackage[doublespacing]{setspace}
\usepackage{rotating}
\usepackage[header,page]{appendix}

\usepackage{color}
\usepackage{xcolor}
\usepackage{colortbl}
\definecolor{mygray}{gray}{.3}
\usepackage[bookmarks=true, pdftex, bookmarksnumbered=true ,bookmarksopen=true
,pdfauthor={Hongbiao Zhao}, pdfstartview=FitB, colorlinks=true ,linkcolor=mygray
,citecolor=mygray,urlcolor=mygray]{hyperref}

\usepackage{graphicx}
\usepackage{epstopdf}
\usepackage{array}
\usepackage{tabularx}
\usepackage{hhline}
\usepackage{tocbibind}

\usepackage[hang,footnotesize,bf]{caption}
\usepackage{float}
\usepackage[margin=1in]{geometry}

\usepackage{diagbox}
\providecommand{\U}[1]{\protect\rule{.1in}{.1in}}

\usepackage{amsthm}
\numberwithin{equation}{section}
\theoremstyle{theorem}\newtheorem{definition}{Definition}[section]
\theoremstyle{definition}
\theoremstyle{definition}
\theoremstyle{definition}
\theoremstyle{theorem}
\theoremstyle{theorem}\newtheorem{theorem}{Theorem}[section]
\theoremstyle{theorem}\newtheorem{corollary}{Corollary}[section]
\theoremstyle{theorem}
\theoremstyle{theorem}\newtheorem{proposition}{Proposition}[section]
\theoremstyle{theorem}
\theoremstyle{remark}
\theoremstyle{remark}
\theoremstyle{remark}
\theoremstyle{remark}
\theoremstyle{remark}
\theoremstyle{remark}

\usefont{T1}{tnr}{m}{sl}
\newcolumntype{R}{>{\raggedleft\arraybackslash}X}
\newcolumntype{L}{>{\raggedright\arraybackslash}X}

\newenvironment{keywords}[1][Keywords]{\noindent\textbf{#1:} }{\noindent}
\newenvironment{jel}[1][JEL Classification]{\noindent\textbf{#1:} }{\noindent}
\newenvironment{msc}[1][Mathematics Subject Classification (2010)]{\noindent\textbf{#1:} }{\noindent}

\allowdisplaybreaks

\begin{document}

\title{\textbf{A Two-Phase Dynamic Contagion Model for COVID-19}}

\author{
Zezhun Chen
\\{\textit{London School of Economics}}
\and
Angelos Dassios
\\{\textit{London School of Economics}}
\and
Valerie Kuan
\\{\textit{ University College London}}
\and
Jia Wei Lim
\\{\textit{Brunel University London}}
\and
Yan Qu\thanks{Correspondence Email: y.qu3@lse.ac.uk}
\\{\textit{University of Warwick}}
\and
Budhi Surya
\\{\textit{Victoria University of Wellington}}
\and
Hongbiao Zhao
\\{\textit{Shanghai University of Finance and Economics}}
}

\date{\today}
\maketitle

\begin{abstract}
In this paper, we propose a continuous-time stochastic intensity model, namely, {\it two-phase dynamic contagion process} (2P-DCP), for modelling the epidemic contagion of COVID-19  and investigating the lockdown effect based on the dynamic contagion model introduced by \citet{dassios2011dynamic}. It allows randomness to the infectivity of individuals rather than a constant reproduction number as assumed by  standard models. Key epidemiological quantities, such as the distribution of final epidemic size and expected epidemic duration, are derived and estimated based on real data for various regions and countries. The associated time lag of the effect of  intervention in each country or  region is estimated. Our results are consistent with the incubation time of COVID-19 found by recent medical study. We demonstrate that our model could potentially be a valuable tool in the modeling of COVID-19. More importantly, the proposed model of 2P-DCP could also be used as an important tool in epidemiological modelling as this type of contagion models with very simple structures is adequate to describe the evolution of regional epidemic and worldwide pandemic.
\end{abstract}

\begin{keywords}
{\footnotesize {Stochastic intensity model; Stochastic epidemic model; Two-phase dynamic contagion process; COVID-19; Lockdown} }
\end{keywords}
\\
\begin{msc}
{\footnotesize {Primary: 60G55; Secondary: 60J75} }
\end{msc}
\\
\begin{jel}
{\footnotesize {I1; H0; E6} }
\end{jel}

\section{Introduction}

In the early stages of epidemic modelling, the spread of diseases was formulated as a deterministic process. The classical deterministic model of {\it susceptible-infectious-recovered} (SIR) was introduced in the seminal paper of \citet{kermack1927contribution}. It models the spread and ultimate containment of an infection in a setting where those who recover are immune to the disease and thus the susceptible population declines over time. Many epidemic models are variations of the SIR model, see \citet{brauer2008mathematical,keeling2008modeling,diekmann2013mathematical} and \citet{martcheva2015introduction}. For example, during the outbreak of COVID-19 since December 2019, a commonly adopted approach for predicting the number of infections is the {\it susceptible-exposed-infectious-recovered} (SEIR) model, which adds an exposed period to the SIR model for accounting the reported incubation period of COVID-19 during which individuals are not yet infectious, e.g. \citet{berger2020seir,liu2020reproductive} and \citet{tian2020investigation}. More recently, \citet{acemoglu2020multi} develop a multi-risk SIR model, which takes into account that different subpopulations have different risks and is applied to analysing optimal lockdown.  \\

However, the random nature of the epidemics spread in our real world suggests that a stochastic model is needed. A continuous-time stochastic counterpart of SIR model was first proposed by \citet{mckendrick1925applications}, and then a variety of stochastic models were studied in the literature, e.g. \citet{bartlett1949some,bartlett1956deterministic} and \citet{bailey1950simple,bailey1953total,bailey1957mathematical}. For more recent developments on stochastic epidemic models in general, see \citet{daley1999epidemic,andersson2000stochastic,allen2008introduction} and \citet{fuchs2013inference}.
In particular, many researchers adopted branching processes. \citet{ball1983threshold} used the birth-and-death process for constructing a sequence of general stochastic epidemics, and \citet{ball1995strong} used branching processes to approximate the early stages of epidemic dynamics, see also \citet{britton2010stochastic} and \citet{ball2016expected}.\\

In this paper, we propose a continuous-time stochastic epidemic model, namely, the {\it two-phase dynamic contagion process} (2P-DCP), for modelling the epidemic contagion. It is a branching process, and can be considered as a generalisation of dynamic contagion process \citep{dassios2011dynamic}, which is an extension of the classical Hawkes process \citep{hawkes1971point,hawkes1971spectra}. In fact, Hawkes process and its various generalisations were originally used for modelling earthquakes in seismology, and recently become extremely popular for modelling financial contagion in economics, see \citet{bowsher2007modelling,large2007measuring,embrechts2011multivariate,bacry2013modelling,
bacry2013some,ait2010modeling,dassios2017exact,dassios2017ageneralised} and \citet{qu2019efficient}.
Analogously, we advocate that they are also applicable to epidemiology. As not all individuals are equally infectious in reality, the main advantage of this Hawkes-based approach is that,  it allows randomness to the infectivity of individuals, rather than a constant reproduction number (the average number of subsequent infections of an infected individual) in standard models.
In this paper, we adopt the 2P-DCP as a more realistic and parsimonious example of Hawkes-based models for modelling the current progression of COVID-19 and investigating the lockdown effect. Key epidemiological quantities, such as the distribution of final epidemic size and expected epidemic duration, have been derived and estimated based on real data. Pandemics have largely shaped the history of human being as described in the popular book by \citet{mcneill1976plagues}, and have made huge impacts to our society and and economy. However, mathematical models developed in epidemiology and economics don't talk to each other much until the current outbreak of COVID-19, which needs urgent calls (e.g. from \href{https://royalsociety.org/news/2020/03/urgent-call-epidemic-modelling/}{The Royal Society}) for researchers across disciplines to work  together and jointly support the scientific modelling for epidemics, see recent intensive interplays between the two fields, e.g. \citet{acemoglu2020multi,alvarez2020simple,atkeson2020will,eichenbaum2020macroeconomics} and \citet{guerrieri2020macroeconomic}. Our paper also responds to the calls by introducing the Hawkes-based approach as a potentially very valuable tool for epidemic modelling.\\

This paper is organised as follows:  Section \ref{Sec_2PDCP} offers the preliminaries including an introduction and formal mathematical definitions    for our stochastic epidemic model,  a two-phase dynamic contagion process.
Key distributional properties, such as the distribution of final epidemic size and expected epidemic duration, are provided in Section \ref{Sec_Distribution}.
In Section \ref{Sec_Implementation}, our model is implemented based on real data, and the associated time lag of the effect of  intervention   in each country or  region is estimated.
Finally, Section \ref{Sec_Conclusion} draws a conclusion for this paper, and proposes some issues for possible further extensions and future research.

\section{Two-Phase Dynamic Contagion Process}\label{Sec_2PDCP}

In this section, we introduce a {\it two-phase dynamic contagion process (2P-DCP)} for modelling the dynamics of COVID-19 contagion.
The  unobservable  {\it effective time} that   aggregated  government interventions (e.g. lockdown of a city or country) came into effect is denoted by $\ell>0$, which divides the COVID-19 epidemic dynamics into two phases. Note that, the time point $\ell$ is different from the exact timing of intervention that can be observed.
The cumulated number of infected individuals is described by a counting process $N_t$ with $N_0=0$, and it is modelled by a two-phase dynamic contagion process defined as below.

\begin{definition}[Two-Phase Dynamic Contagion Process (2P-DCP)]
A two-phase dynamic contagion process (2P-DCP) is a point process $N_t$ with two phases:
\begin{description}
  \item[Phase 1 (Full Contagion):] For the first phase period $t \in [0, \ell)$, $N_t$ follows a {\it dynamic contagion process} with stochastic intensity
  \begin{equation}\label{Eq_DCP_intensity_Phase1}
  \lambda_t
  ~=~
  \lambda_0 e^{-\delta t} + \sum_{k=1}^{N^*_t} Z_k e^{-\delta (t-T_k^*)} + \sum_{i=1}^{N_t}Y_i e^{-\delta (t-T_i)},
  \qquad t \in [0, \ell],
  \end{equation}
  where
  \begin{itemize}
    \item $\lambda_0>0$ is the initial intensity at time $t=0$;

    \item $\delta>0$ is the constant rate of exponential decay;

    \item $N^*_t\equiv\{T^*_k\}_{k=1,...}$ is a Poisson process of constant rate $\varrho>0$ arriving in time $t\le \ell$;

    \item $\{Z_k\}_{k=1,...,N_{\ell}^*}$ are  i.i.d. {\it externally-exciting} jump sizes, realised at times $\{T_k^*\}_{k=1,..., {N_{\ell}^*}}$, with distribution $H(y)$; 

    \item $\{Y_i\}_{i= 1,...,N_{\ell}}$ are i.i.d. {\it self-exciting} jump sizes of the first phase, realised at  times $\{T_i\}_{i=1,..., {N_{\ell}}}$ , with distribution $G_1(y)$.

  \end{itemize}

  \item[Phase 2 (Self Contagion):] For the second phase period $t \in (\ell, \infty)$, $N_t$ is a pure self-exciting {\it Hawkes process}
  with stochastic intensity
  \begin{equation}\label{Eq_DCP_intensity_Phase2}
    \lambda_t
  ~=~
 {\lambda_{\ell}} e^{-\delta {(t-\ell)}} + \sum_{i={N_{\ell} +1}}^{N_t}Y_i e^{-\delta (t-T_i)},
  \qquad
  t \in (\ell, \infty),
  \end{equation}
  where
  \begin{itemize}
    \item ${\lambda_{\ell}}$ is the initial intensity of the second phase starting at the cutoff time point $\ell$, which is the terminal intensity of the first phase;

    \item $\{Y_i \}_{i={N_{\ell}+1},...}$ are i.i.d. {\it self-exciting} jump sizes of the second phase, realised at times $\{T_i\}_{i={N_{\ell}+1},...}$, with distribution $G_2(y)$. 
        Note that compared with the average mean of the self-exciting jump size for Phase 1, the mean of self-exciting jump in Phase 2 is smaller, which demonstrates that the COVID-19 becomes less contagious on average after time point $\ell$.
  \end{itemize}
\end{description}
\end{definition}

\begin{figure}[hbt!]
 \begin{center}
    \includegraphics[width=1\textwidth]{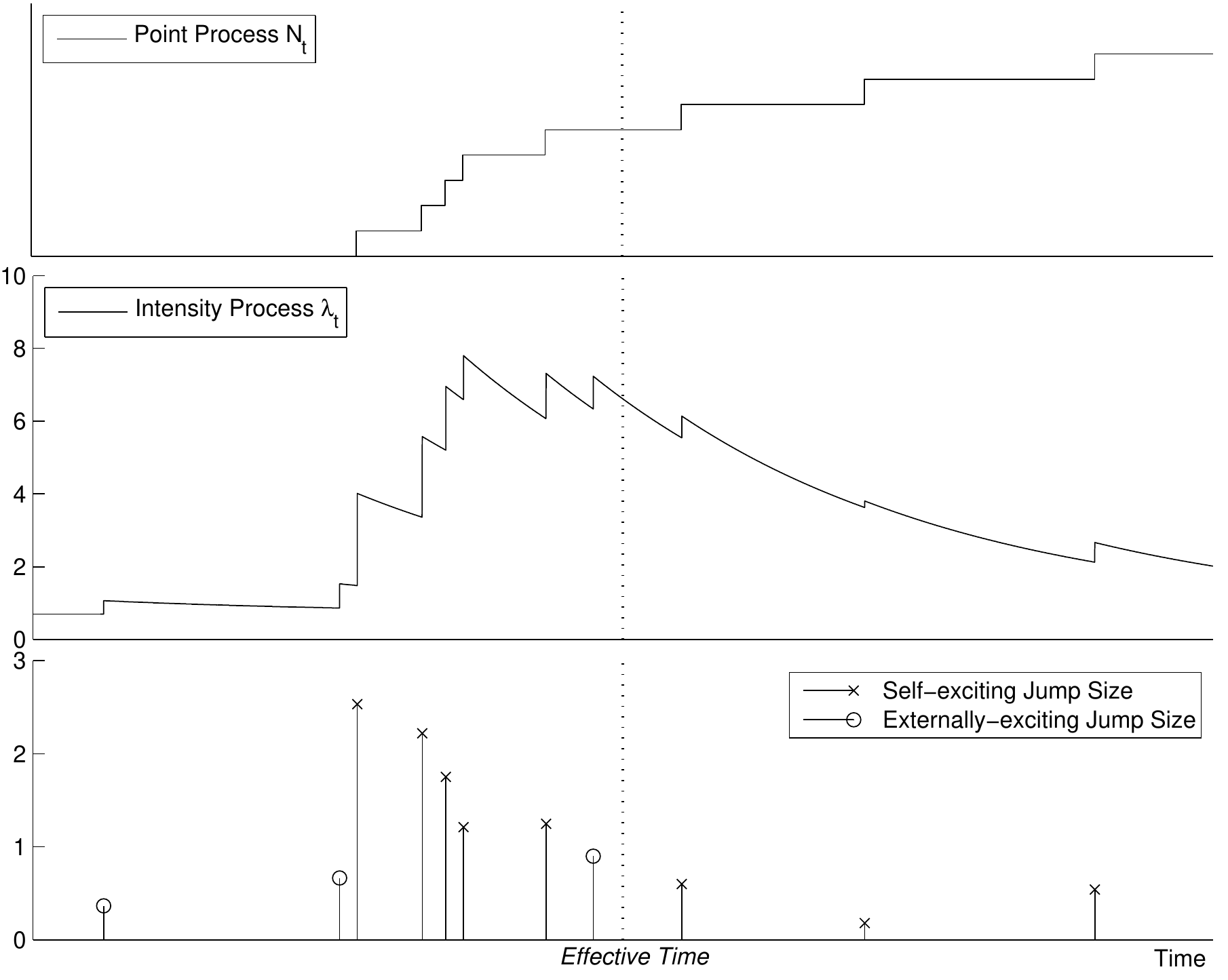}
    \caption{Two-phase dynamic contagion process}
    \label{Fig_plot_N_lambda_2PDCP}
 \end{center}
\end{figure}

The point process $N_t$ and its intensity process $\lambda_t$ are illustrated in Figure \ref{Fig_plot_N_lambda_2PDCP}.
Overall, we can more compactly define our new pandemic model, a two-phase dynamic contagion process, as  a counting process $N_t\equiv\{T_i\}_{i=1,...}$ with $N_0=0$ and stochastic intensity
\begin{equation}\label{Eq_DCP_intensity}
  \lambda_t
  ~=~
  \lambda_0 e^{-\delta t} + \sum_{k=1}^{N^*_t} Z_k \mathbf{1}_{\{t\le \ell \}} e^{-\delta (t-T_k^*)} + \sum_{i=1}^{N_t}Y_i e^{-\delta (t-T_i)},
  \qquad
  t\geq0,
\end{equation}
where
\begin{itemize}
  \item $\{Y_i\}_{i=1,...N_t}$ are i.i.d.  {\it self-exciting} jump sizes with a two-phase distribution $G(y;t)$, i.e.,
\begin{equation}\label{Eq_G}
G(y;t) = G_1(y)\textbf{1}_{\{t\leq \ell\}} + G_2(y)\textbf{1}_{\{t> \ell\}}.
\end{equation}

  \item $\{ Z_k \}_{k=1,...N_t^*}$ are i.i.d.  {\it externally-exciting} jump sizes  with  distribution $H(y)$.
\end{itemize}

This equivalent definition as a dynamic contagion process has an advantage: it can be viewed as a branching process and has a more intuitive interpretation with regard to a pandemic. The cluster-process presentation is provided as follows.
\begin{itemize}
  \item The cumulated number of infected cases, $N_t$, is a cluster point process, which consists of two types of points: {\it outside-imported} cases and  {\it inside-infected} cases.

  \item The arrivals of {\it outside-imported} cases follows a Cox process with shot-noise intensity
\[
\lambda_0 e^{-\delta t} + \sum_{k=1}^{N^*_t} Z_k  e^{-\delta (t-T_k^*)},
\]
where externally-exciting jumps arrive as a Poisson process $N^*_t$ at time points $\{T^*_k\}_{k=1,...}$ with sizes (marks) $\{Z_k\}_{k=1,...}$, and they disappear after time point $\ell$ when the interventions took effect, and there will be no any increase of imported cases in a long run.

  \item Each imported case may infect other individuals inside and thereby causes new cases, and each of these new cases would further infect others inside, and so on. The infection of any {\it new} cases caused by the {\it previous} infected cases follows a Cox process with exponentially decaying intensity $Y_{\cdot}e^{-\delta (t-T_{\cdot})}$, where $Y_{\cdot}$ follows a two-phase distribution $G(y;t)$ and $T_{\cdot}$ is the infection time of the {\it previous} infected case. After the interventions took effect, the COVID-19 becomes less easy to spread on average. This is captured by our assumption of two-phase distribution (\ref{Eq_G}) for $Y_{\cdot}$ here.

  \item Overall, the superposition of all these infected cases form a point process $N_t$, a two-phase dynamic contagion process with stochastic intensity (\ref{Eq_DCP_intensity}).
\end{itemize}

\section{Distributional Properties}\label{Sec_Distribution}
In this section, we outline key distributional properties for the two-phase dynamic contagion process.  We derived the conditional joint Laplace transform of $\lambda_t$ and probability generating function of $N_t$, which are the key results to further derive the elimination probability of the epidemic and the distribution of the final epidemic size.

\subsection*{Joint Distribution of   $(\lambda_t, N_t)$}
Let $\{\mathcal{F}_t\}_{t\ge 0}$ be the natural filtration of the point process $N_t$, i.e. $\mathcal{F}_t=\sigma\left(N_s, s\le t \right)$ and assume that the intensity process $\{\lambda_t\}_{t\ge 0}$ being $\mathcal{F}_t$-adapted. The joint Laplace transform and probability generating function for $(\lambda_t, N_t)$ is provided in Theorem \ref{thm_joint_lf_Nt_and_lambda_t} as below.

\begin{theorem}\label{thm_joint_lf_Nt_and_lambda_t}
 For time $s\le t$, the conditional joint Laplace transform and probability generating function for $\lambda_t$ and the point process $N_t$ is of the form
\begin{eqnarray}\label{eq_conditional_lf_pgf}
\mathbb{E}\left[ \theta^{N_t} e^{-v \lambda_t}|\mathcal{F}_s \right]
=
\begin{cases}
\theta^{N_s}e^{c(s)}e^{-A(s)\lambda_s},&\quad 0\le s\le t \le \ell,\\
 \theta ^ {N_s} e^{-A(s)\lambda_s}, &\quad \ell<s\le t,
\end{cases}
\end{eqnarray}
where $A(s)$ is determined by the nonlinear ordinary differential equation (ODE)
\begin{equation}\label{eq_ode_lf_pgf}
A'(s) -\delta A(s) +1 -\theta \hat{g}(A(s);s)=0,
\end{equation}
where the boundary condition is $A(t)=v$ with
\[
\hat{g}(u;t) = \int\limits_0^\infty e^{-uy}\mathrm{d}G(y;t),
\]
and $c(t)$ is determined by
\begin{equation}\label{eq_ode_lf_pgf_before_intervention}
c(t)=\varrho\mathbf{1}_{\{t\leq \ell\}}\int\limits_0^t\left[ 1-\hat{h}(A(u)) \right]\mathrm{d}u,
\end{equation}
with
\[
\hat{h}(u)=\int\limits_0^\infty e^{-uy}\mathrm{d}H(y).
\]
\end{theorem}

\begin{proof}
For $0\le s\le t\le \ell$,   $\lambda_t$ is the intensity process of the dynamic contagion process introduced in \citet{dassios2011dynamic}. The corresponding conditional joint Laplace transform, probability generating function for the process $\lambda_t$ and the point process $N_t$ is provided in Theorem 3.1 of \citet{dassios2011dynamic}. For $\ell<s\le t$, given $\mathcal{F}_s$, the infinitesimal generator of the dynamic contagion process $(\lambda_s, N_s, s)$ acting on a function $f(\lambda,n, s)$ within its domain $\Omega(\mathcal{A})$ is given by
\begin{equation}\label{eq_generator_phase2}
\mathcal{A}f(\lambda, n, s)=\frac{\partial f}{\partial s}-\delta \lambda\frac{\partial f}{\partial \lambda}+\lambda\left(\int\limits_0^\infty f(\lambda+y, n+1, s)\mathrm{d}G(y;s) -f(\lambda,n, s) \right).
\end{equation}
Consider a function $f(\lambda, n, s)$ of form
\[
f(\lambda, n , s) = \theta ^ n e^{-A(s)\lambda},
\]
and substitute this into $\mathcal{A}f=0$, we then have the  ODE
\[
A'(s) =\delta A(s) -1 +\theta \hat{g}\big(A(s),s\big),
\]
adding the boundary condition $A(t)=v$, gives the ODE in \eqref{eq_ode_lf_pgf}.
\end{proof}

The moments of $\lambda_t$ and $N_t$ can be obtained by differentiating the joint Laplace transform and probability generating function of $\lambda_t$ and $N_t$, and the results are provided in Proposition \ref{eq_condtional_expection_lambda_n} and \ref {eq_condtional_expection2_lambda_n} as below.

\begin{proposition}\label{eq_condtional_expection_lambda_n}
The conditional expectation of the process $\lambda_t$ given  $\mathcal{F}_s$ for $s\le t$ is given by
\begin{eqnarray}\label{eq_conditional_expectation_on_lambda_t}
\mathbb{E}\left[\lambda_t|\mathcal{F}_s \right]=
\begin{cases}
 \frac{\varrho\mu_{H} \mathbf{1}_{\{t\le \ell \} }}{\kappa }+\left( \lambda_s -  \frac{\varrho\mu_{H} \mathbf{1}_{\{t\le \ell \} }}{\kappa  } \right) e^{-\kappa(t-s)},  &\quad \kappa \neq 0,  \\
 \lambda_s+\varrho\mu_{H} \mathbf{1}_{\{t\le \ell \} }(t-s), &\quad  \kappa= 0,
 \end{cases}
\end{eqnarray}
and the conditional expectation of the point process $N_t $ given $\mathcal{F}_s$ is of the form
\begin{eqnarray}\label{eq_conditional_expectation_Nt}
\mathbb{E}\left[N_t|\mathcal{F}_s\right]=
\begin{cases}
N_s+ \frac{\varrho\mu_{H} \mathbf{1}_{\{t\le \ell \} }(t-s)}{\kappa }+\left( \lambda_s -  \frac{\varrho\mu_{H} \mathbf{1}_{\{t\le \ell \} }}{\kappa } \right) \frac{1-e^{-\kappa(t-s)}}{ \kappa },  &\quad \kappa \neq 0,  \\
 N_s+\lambda_s(t-s)+\frac{1}{2}\varrho\mu_{H} \mathbf{1}_{\{t\le \ell \} }(t-s)^2, &\quad  \kappa= 0,
 \end{cases}
\end{eqnarray}
where
\[
\mu_H=\int\limits_0^\infty y\mathrm{d}H(y),\quad \mu_G=\int\limits_0^\infty y\mathrm{d}G_1(y) \mathbf{1}_{\{t\le  \ell\}} +\int\limits_0^\infty y\mathrm{d}G_2(y)\mathbf{1}_{\{t> \ell\}},
\]
and $\kappa=\delta-\mu_G$.
\end{proposition}
\begin{proof}
The result in \eqref{eq_conditional_expectation_on_lambda_t} immediately follows Theorem 3.6 in \citet{dassios2011dynamic}. And since $N_t-N_s-\displaystyle \int\limits_s^t\lambda_u\mathrm{d}u$ is a martingale, we have
\[
\mathbb{E}\left[ N_t-N_s|\mathcal{F}_s\right]=\int\limits_s^t \mathbb{E}\left[ \lambda_u|\mathcal{F}_s\right] \mathrm{d}u,
\]
which directly implies the result in \eqref{eq_conditional_expectation_Nt}.
\end{proof}

\begin{proposition}\label{eq_condtional_expection2_lambda_n}
The conditional second moment of the process $\lambda_t$ given $\mathcal{F}_s$ for $s \le t $ is given by
\begin{eqnarray}\label{eq_conditional_second_on_lambda_t}
\mathbb{E}\left[\lambda^2_t|\mathcal{F}_s \right]=
\begin{cases}
 \lambda_s^2 e^{-2\kappa t} + \frac{2\varrho\mu_{H} \mathbf{1}_{\{t\le \ell \} } + \mu_{2_G}}{\kappa} \left( \lambda_s - \frac{\varrho  \mu_{H}}{\kappa} \right)  \left( e^{-\kappa (t-s)} - e^{-2\kappa (t-s)} \right) \\ \quad + \left( \frac{(2\varrho\mu_{H} + \mu_{2_G})\varrho\mu_{H}  \mathbf{1}_{\{t\le \ell \} }}{2\kappa^2} +  \frac{\varrho\mu_{2_H} \mathbf{1}_{\{t\le \ell \} } }{2\kappa } \right) \left(1-e^{-\kappa(t-s)} \right),  &\quad \kappa\neq 0,  \\
 \lambda^2_s+ \lambda_s\mu_{2_G}t +  \left( 2\lambda_0\mu_{H} + \mu_{2_H}\right)\varrho\mathbf{1}_{\{t\le \ell \} }(t-s) \\  \quad+ \left( \varrho^2\mu_{H}^2 + \frac{1}{2}\varrho\mu_{H}\mu_{2_G}\right) \mathbf{1}_{\{t\le \ell \} }(t-s)^2, &\quad  \kappa =0,
 \end{cases}
\end{eqnarray}
and the conditional joint expectation of the process $\lambda_t$ and the point process $N_t $ given $\mathcal{F}_s$ for $s > \ell $ is of the form
\begin{eqnarray}\label{eq_joint_conditional_expectation_lambdat_Nt}
\mathbb{E}\left[\lambda_t N_t|\mathcal{F}_s\right]=
\begin{cases}
  \lambda_s N_s e^{-\kappa (t-s)} + \left( \lambda_s\mu_{G} + \frac{\lambda_s\mu_{2_G}}{\kappa} \right) (t-s) e^{-\kappa (t-s)} \\ \quad + \left(\frac{\lambda_s^2}{\kappa} - \frac{\lambda_s \mu_{2_G}}{\kappa} \right)(e^{-\kappa (t-s)} -e^{-2\kappa (t-s)}), &\quad \kappa\neq 0, \\ \lambda_s N_s +
 \left( \lambda_s^2 + \lambda_s\mu_{G} \right)(t-s) + \frac{1}{2}\lambda_s\mu_{2_G}(t-s)^2, &\quad  \kappa=0
 \end{cases}
\end{eqnarray}
and the second moment of point process $N_t $ given $\mathcal{F}_s$ for $s > \ell $ is of the form
\begin{eqnarray}\label{eq_joint_conditional_expectation_lambdat_Nt}
\mathbb{E}\left[N_t^2|\mathcal{F}_s\right]=
\begin{cases}
  \left(\frac{\lambda_s\mu_{2_G}}{\kappa^3}-\frac{\lambda_s^2}{\kappa^2} \right) (1-e^{-2\kappa (t-s)}) - \left(\frac{2\lambda_s\mu_{G}}{\kappa} + \frac{2\lambda_s\mu_{G_2}}{\kappa^2}\right) (t-s)e^{-\kappa(t-s)} \\ \quad + \left( \frac{\lambda_s+2\lambda_s N_s}{\kappa} + \frac{2\lambda_s^2 + 2\lambda_s\mu_{2_G}}{\kappa^2} \right)(1-e^{-\kappa (t-s)}), &\quad \kappa\neq 0,  \\
 \lambda_s(\mu_{G} + 2N_s)(t-s) + \left( \lambda_s^2 + \lambda_s\mu_{G}\right) (t-s)^2 + \frac{1}{3}\lambda_s\mu_{2_G} (t-s)^3, &\quad  \kappa=0,
 \end{cases}
\end{eqnarray}
where
\[ 
\mu_{2_H}=\int\limits_0^\infty y^2\mathrm{d}H(y),\quad \mu_{2_G}=\int\limits_0^\infty y^2\mathrm{d}G_1(y)\mathbf{1}_{\{t\le \ell\}} +\int\limits_0^\infty y^2\mathrm{d}G_2(y)\mathbf{1}_{\{t> \ell\}},
\]
and $\kappa = \delta - \mu_{G}$.
\end{proposition}

\begin{proof}
These results immediately follows Lemma 3.1, Lemma 3.2 and Theorem 3.9 in \citet{dassios2011dynamic}.
\end{proof}

\subsection*{Probability of Elimination Time $\widetilde{T}$ }
After the government interventions come into effect, the contagion rate will dramatically decline and new cases will drop abruptly almost to  nothing in the near future. It is therefore of great interests to calculate the probability of  elimination time, i.e. the time that the last ever case  arrives, after government interventions come into effect. Let $\widetilde{T}$ to be {\it elimination time} such that
\begin{equation}
\widetilde{T}
:=
\inf\big\{ t> \ell : \forall s\ge t, \quad N_s-N_t=0 \big\}.
\end{equation}
The condition probability of the elimination time is provided in Proposition
\ref{prop_distribution_of_elimination_time}.

\begin{proposition}\label{prop_distribution_of_elimination_time}
For $\ell\le s\le t$, the elimination probability is given by
\begin{equation}
\mathbb{P}\left( \widetilde{T}\le t|\mathcal{F}_s\right) = e^{-A(s)\lambda_s},
\end{equation}
where $A(s)$ is determined by the ODE in \eqref{eq_ode_lf_pgf} with boundary condition $A(t)=\frac{1}{\delta}$.
\end{proposition}

\begin{proof}
Given $\widetilde{T}$ being the timing of the last ever event,  the event $\{\widetilde{T}\le t\}$ implies that $N_u=N_t$ for any $u\ge t$, which also lead to
\[
\lambda_u=e^{-\delta(u-t)}\lambda_t.
\]
Hence, we have
\begin{eqnarray}
\mathbb{P}\left( \widetilde{T}\le t|\mathcal{F}_s\right)
&=&
\mathbb{E}\left[ \mathbf{1}_{\{ \widetilde{T}\le t\}}\mid \mathcal{F}_s \right]\nonumber\\
&=&
\mathbb{E}\left[ \exp\left(-\int\limits_t^\infty \lambda_t e^{-\delta(u-t)}\mathrm{d}u \right) \mid \mathcal{F}_s \right]\nonumber\\
&=&
\mathbb{E}\left[  \exp\left( -\frac{\lambda_t}{\delta}\right) \mid \mathcal{F}_s\right]
\end{eqnarray}
And according to Theorem  \ref{thm_joint_lf_Nt_and_lambda_t}, by setting $\theta=1$ in \eqref{eq_conditional_lf_pgf}, the result follows immediately.
\end{proof}

\subsection*{Joint Expectation of Epidemic Size $N_t$ and Elimination Time $\widetilde{T}$ }
Given the last ever event $\{ \widetilde{T}<t\}$, one could obtain the expected size of the epidemic at time $t$. The relevant details are presented in Corollary \ref{prop_expectation_epidemic_size}.

\begin{corollary}\label{prop_expectation_epidemic_size}
For $\ell\le s\le \widetilde{T}\le t$, the conditional joint expectation of $N_t$  and $\mathbf{1}_{\{\widetilde{T}\le t\}}$ is of the following form
 \begin{eqnarray}\label{eq_expectation_of_epidemic_size}
\mathbb{E}\left[N_t\mathbf{1}_{\{\widetilde{T}\le t\}} \mid \mathcal{F}_s \right] = \frac{\mathrm{d} } {\mathrm{d}\theta} \left\{\theta^{N_s} e^{-A(s)\lambda_s} \right\} \bigg|_{\theta=1^-},
\end{eqnarray}
where $A(s)$ satisfies the ODE in \eqref{eq_ode_lf_pgf} with boundary condition $A(t)=\frac{1}{\delta}$.

\end{corollary}

\begin{proof}
According to Theorem \ref{thm_joint_lf_Nt_and_lambda_t} and Proposition \ref{prop_distribution_of_elimination_time}, we have
\begin{eqnarray}
\mathbb{E}\left[ \theta^{N_t} \mathbf{1}_{\{\widetilde{T}\le t\} } \mid \mathcal{F}_s\right]
&=&
  \mathbb{E}\left[\theta^{N_t} e^{-\frac{\lambda_t}{\delta}} \mid\mathcal{F}_s \right]  \nonumber\\
&=&
 \theta^{N_s} e^{-A(s)\lambda_s}.
\end{eqnarray}
Since $\mathbb{E}\left[N_t\mathbf{1}_{\{\widetilde{T}\le t\}} \mid \mathcal{F}_s \right] = \frac{\mathrm{d}}{\mathrm{d}\theta} \mathbb{E}\left[ \theta^{N_t} \mathbf{1}_{\{\widetilde{T}\le t\} } \mid \mathcal{F}_s\right] \bigg |_{\theta=1^{-}}$, the result immediately follows \eqref{eq_expectation_of_epidemic_size}.

\end{proof}

\subsection*{Distribution of Final Epidemic Size $N_{\infty}$}
The final epidemic size is one of the most important  epidemiological quantities to study. In fact, under the two-phase dynamic contagion model, the final epidemic size is the value of the point process $N_t$ when time goes to infinity. Conditional on $s>\ell$, since there are no externally-exciting jumps in the intensity, the distribution of $N_{\infty}$ can be characterised by  Proposition \ref{prop_final_epidemic_size_distribution} as below.
\begin{proposition}\label{prop_final_epidemic_size_distribution}
For $\ell<s$, the probability generating function of $N_{\infty}$ conditional on $\mathcal{F}_s$ is given as
\begin{equation}
\mathbb{E}\left[ \theta^{N_\infty}\mid \mathcal{F}_s\right]
=
e^{-v^*\lambda_s},
\end{equation}
where
\[
v^* = \frac{1}{\delta}\left(1-\theta \int\limits_{0}^\infty e^{-v^*y}\mathrm{d}G_2(y) \right).
\]
\end{proposition}
\begin{proof}
The result immediately follows Theorem 3.5 in \citet{dassios2011dynamic}.
\end{proof}
While the government interventions come into effect, if we assume  {\it i.i.d.} self-exciting jump sizes $Y_i\sim\textrm{Exp}(\beta)$ for $i=N_{\ell}+1,.... $, then, we have an explicit expression for the probability generating function of $N_{\infty}$ as
\[
\mathbb{E}[\theta^{N_{\infty}} \mid \mathcal{F}_s ]
=
\exp\left( - \frac{ \sqrt{(\delta \beta -1)^2 + 4 \delta \beta(1-\theta)} - (\delta \beta -1) }{2\delta} \lambda_{s} \right), \quad s>\ell.
\]
This implies that, the final epidemic size $N_{\infty}$ conditional on $\mathcal{F}_{\ell}$ follows a mixed-Poisson distribution with the probability mass function
\begin{equation}\label{eq_mix_poisson_distribution}
\mathbb{P}\left( N_{\infty} = k \mid \mathcal{F}_s\right)
=
\int\limits^\infty_0 \frac{v^k e^{-v}}{k!} m(v) \mathrm{d}v,
\qquad k =0,1,....,
\end{equation}
where $m(v)$ is the density function of the mixing distribution,
\[
m(v):=
\exp\left(\frac{\delta \beta -1}{2 \delta} \lambda_{s}
-  \left( \frac{\delta \beta -1}{2 \delta} \right)^2  \frac{\delta}{\beta} v
- \frac{ \frac{\beta}{2\delta} \lambda_{s}^2 }{2v}
 \right)
\frac{\sqrt{\frac{\beta}{2 \delta}} \lambda_{s}}{\sqrt{2 \pi}v^{\frac{3}{2}}},
\]
which is an inverse Gaussian distribution with parameters $\frac{\beta}{\delta \beta - 1}\lambda_{s}$ and $\frac{\beta}{2\delta}\lambda_{s}^2$.

\section{Empirical Study}\label{Sec_Implementation}

We provide a calibration scheme based on the daily increments of the two-phase dynamic contagion process $N_t$. Let us first denote the observations of the daily confirmed COVID-19 cases as $\{C_t\}_{t=0,1,2,...,T}$ . The mean square error (MSE) between the expected daily increments of $N_t$ and the actual reported daily confirmed COVID-19 cases is given as
\begin{equation}\label{eq_MSE}
\textsf{MSE}(\alpha,\beta, \delta, \varrho, \ell)
=
\frac{1}{T}\sum\limits_{t=0}^{T-1}\bigg(\mathbb{E}\left[ N_{t+1}-N_t  \right]-C_{t+1} \bigg)^2.
\end{equation}
We consider the calibration based on minimising the MSE  \eqref{eq_MSE}, i.e., we choose parameters $(\hat{\alpha}, \hat{\beta}, \hat{\delta},\hat{\varrho}, \hat{\ell})$ such that
\[
\textsf{MSE}(\hat{\alpha}, \hat{\beta}, \hat{\delta},\hat{\varrho}, \hat{\ell})
~:=~
\min_{(\alpha, \beta, \delta, \varrho, \ell)   }  \textsf{MSE}(\alpha, \beta, \delta, \varrho, \ell),
\]
with $\alpha,\beta,\delta,\varrho\ge 0$ and $\ell\in\mathbb{N}^+$.\\

Without loss of generality, for simplicity, we assume $Z_k = 1$ for any $k$, $Y_i\sim\textrm{Exp}(\alpha)$ for $i=1,..., N_{\ell}$ and $Y_i\sim\textrm{Exp}(\beta)$ for $i=N_{\ell}+1,... $, and $\lambda_0 = 0$ in \eqref{Eq_DCP_intensity} for model calibration. Other assumptions for $Z_k, \{Y_i\}_{i=1,...N_{\ell}, N_{\ell}+1,...}, \lambda_0$ can also be used if necessary.
We provide two empirical examples. The first one  concentrates on the COVID-19 pandemic in mainland China during the period from early January to late March 2020. The second one focuses on the worldwide COVID-19 pandemic during the period from mid February to early May 2020. The data we used are publicly available. Datasets are mostly cited from the associated official government health department websites for non-European countries and from the European Centre of Disease Control (ECDC) for European countries.\\

\subsection*{COVID-19 Pandemic in  Mainland China}
The daily confirmed COVID-19 cases for regions in China can be obtained from daily reports of the National Health Commission of the PRC.
We use  the reported daily confirmed cases for several regions in China from 2020-01-19 to 2020-03-31 as examples for model calibration. The corresponding estimation results are illustrated in Table  \ref{T_parameter_estimation_regions_China}. We can see that the estimator $\hat{\varrho}$ are quite different from each other. In particular, these regions that are close to Hubei, namely Henan, Hunan, Anhui, have relatively larger intensities for externally-exciting jumps, which means that these regions experienced more external shocks from Wuhan and these external shocks can be originally infected individuals from Wuhan. Naturally,  without taking into account Wuhan, the rest of Hubei has the largest intensity $\hat{\varrho}$ for externally-exciting jumps. The main government intervention established by the Chinese authority is the announcement of the completely lockdown of Wuhan and later the whole Hubei Province on $23$ January $2020$. One or two days later, all other regions enforced the quarantine and raised the alert of public health emergency. Setting the date 2020-01-19 as the initial time $t=0$, then  the government intervention took place when $t =4$ and came into effect when $t=\hat{\ell}$. Since the delay period of the government interventions is the difference of the date when the government interventions came into effect and the date when the government introduced the restriction measures, we can therefore observe from Table \ref{T_parameter_estimation_regions_China} that the estimated delays of the government interventions  for different regions therefore are between $5$ and $15$ days, which are  consistent to the {\it incubation time} of COVID-19 for most people, e.g. as found in a highly cited medical study of \citet{lauer2020incubation}.\\

\begin{table}[hbt!]\footnotesize 
\tabcolsep 0pt  \vspace*{-5pt}
\begin{center}
\def\temptablewidth{1\textwidth}\caption{Calibration parameters $(\hat{\alpha}, \hat{\beta}, \hat{\delta},\hat{\varrho}, \hat{\ell})$ for total confirmed COVID-19 cases from 2020-01-19 to 2020-03-31 for various regions in China.  }
{\rule{\temptablewidth}{1pt}}
\begin{tabular*}{\temptablewidth}{@{\extracolsep{\fill}}c|ccccc|c}
\hline
\backslashbox{~~Regions~~}{~~Parameters~~}& $\hat{\alpha}$ & $\hat{\beta}$  & $\hat{\delta}$ & $\hat{\varrho}$   &$\hat{\ell}$& $\frac{\textsf{MSE}(\hat{\alpha}, \hat{\beta}, \hat{\delta},\hat{\varrho}, \hat{\ell})}{N}$ \\ \hline
 \hline
Heilongjiang & 2.431123 & 8.036146 & 0.252324 & 0.934574 & 16  & 0.036414  \\ \hline
Sicuan       & 5.509284 & 7.402851 & 0.21608  & 3.913496 & 11 & 0.023037  \\ \hline
Shandong     & 5.303477 & 7.868065 & 0.340273 & 5.296465 & 19  & 0.033823  \\ \hline
Jiangxi      & 2.936512 & 7.903677 & 0.285922 & 3.782104 & 15 & 0.070023  \\ \hline
Anhui        & 3.644156 & 7.773938 & 0.339183 & 7.161244 & 18  & 0.037944  \\ \hline
Hunan        & 4.467952 & 7.660323 & 0.293860 & 8.351008 & 15  & 0.053281  \\ \hline
Zhejiang     & 1.398741 & 8.065125 & 0.260449 & 1.627822 & 09 & 0.144322  \\ \hline
Henan        & 3.370932 & 7.867786 & 0.286357 & 6.669105 & 15 & 0.042004  \\ \hline
Guangdong    & 2.316879 & 7.987904 & 0.256143 & 2.360507 & 12  & 0.073389  \\ \hline
Hubei       & 4.468435 & 28.46780  & 0.173874 & 58.24001 & 16  & 0.654813  \\
\hline
\end{tabular*}\label{T_parameter_estimation_regions_China}
{\rule{\temptablewidth}{1pt}}
\end{center}
\end{table}

The {\it branching ratio} (BR), which demonstrates the average infection rate, is determined by $\mathbb{E}[Y_i]/\delta$. In Table \ref{T_branching_ratio_China_regions}, we compare the estimated branching ratios before and after the government interventions came into effect, namely $R_b$ and $R_a$, respectively.
It is clear that the branching ratio for every region decreases significantly when the state changed, i.e. government interventions came into effect.  One can also access the efficiency for when regions implemented the restriction packages introduced by the central government by comparing the corresponding branching ratios $R_b$ and $R_a$. The comparison of $R_b$ and $R_a$ for regions in China are presented in  Figure \ref{Fig_BR_China}. We can see that the government restriction packages had been well-implemented for all regions in China. In particular, Hubei, where the strictest measure, i.e. the completely lockdown of Wuhan and Hubei, had been introduced, shows a dramatic drop of contagion rate after the interventions came into effect.
\\
 \begin{table}[hbt!] \footnotesize 
\tabcolsep 0pt  \vspace*{-5pt}
\begin{center}
\def\temptablewidth{0.8\textwidth}\caption{The estimated branching ratio (BR) before and after the government interventions came into effect, namely $R_b, R_a$ respectively, for regions in China. }
{\rule{\temptablewidth}{1pt}}
\begin{tabular*}{\temptablewidth}{@{\extracolsep{\fill}}c|c c}
\hline
\backslashbox{~~Regions~~}{~~BR~~}    & $R_b$      & $R_a$       \\
 \hline
 \hline
Heilongjiang & 1.630176 & 0.493166                  \\
 \hline
Sicuan       & 0.840021 & 0.625153                   \\
\hline
Shandong     & 0.554130 & 0.373512                     \\
 \hline
Jiangxi      & 1.191025 & 0.442510                     \\
 \hline
Anhui        & 0.809039  & 0.379250                   \\
 \hline
Hunan        & 0.761641   & 0.444234                   \\
\hline
Zhejiang     & 2.744989  & 0.476066                    \\
 \hline
Henan        & 1.035957  & 0.443853                  \\
 \hline
Guangdong    & 1.685052  & 0.488747                     \\
\hline
Hubei        & 1.287094 & 0.202028                 \\
\hline
\end{tabular*}\label{T_branching_ratio_China_regions}
{\rule{\temptablewidth}{1pt}}
\end{center}
\end{table}

\begin{figure}[hbt!]
	\begin{center}
		\includegraphics[width=1\textwidth]{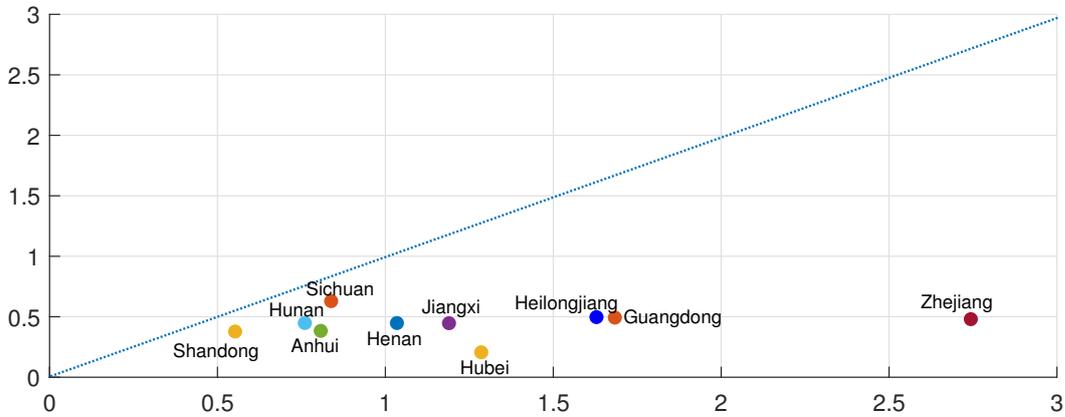}
		\caption{Comparison of the branching ratios {\it before} and {\it after} the government interventions came into effect for regions in China. The horizontal axis represents the branching ratio {\it before} the government interventions came into effect, namely $R_b$ and the vertical axis represents the branching ratio {\it after} the government interventions came into effect, i.e. $R_a$.  }
		\label{Fig_BR_China}
	\end{center}
\end{figure}

Figure \ref{Fig_Comparison_Model_Calibration_China} and \ref{Fig_Total_Cases_China} demonstrate comparisons between the expected daily/total confirmed cases for the two-phase dynamic contagion model under the calibrated parameter $(\hat{\alpha}, \hat{\beta}, \hat{\delta},\hat{\varrho}, \hat{\ell})$ in Table \ref{T_parameter_estimation_regions_China} and the actual daily/total confirmed COVID-19 cases for the period  2020-01-19 to 2020-03-31. We observe that the model allows different shapes of trend before the interventions came into effect. All these regions indicate relatively smooth exponential decay of daily new cases after the peak. In addition, the estimated cumulative confirmed cases are very close to the actual total confirmed COVID-19 cases, which further confirms that our new model is a good candidate for describing the propagation process.

\begin{figure}[hbt!]
\begin{center}
   \includegraphics[width=1\textwidth]{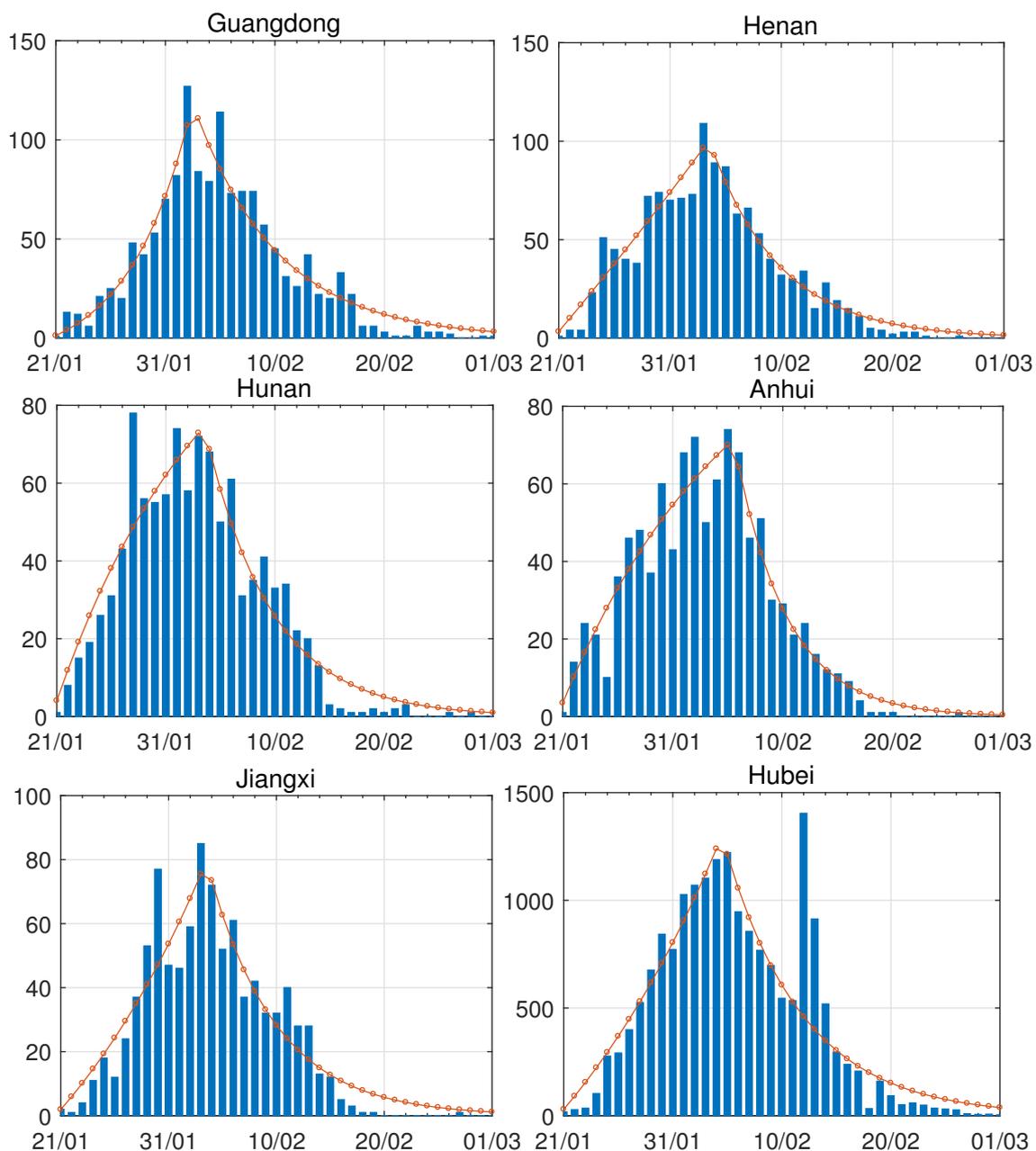}
    \caption{Model calibration comparisons between the expected daily confirmed cases under the calibrated parameters  $(\hat{\alpha}, \hat{\beta}, \hat{\delta},\hat{\varrho}, \hat{\ell})$ in Table \ref{T_parameter_estimation_regions_China} and  actual  daily confirmed COVID-19 cases from 2020-01-19 to 2020-03-31 for Guangdong, Henan, Hunan, Anhui, Jiangxi, and Hubei of China.}
    \label{Fig_Comparison_Model_Calibration_China}
    \end{center}
\end{figure}

\begin{figure}[hbt!]
\begin{center}

   \includegraphics[width=1\textwidth]{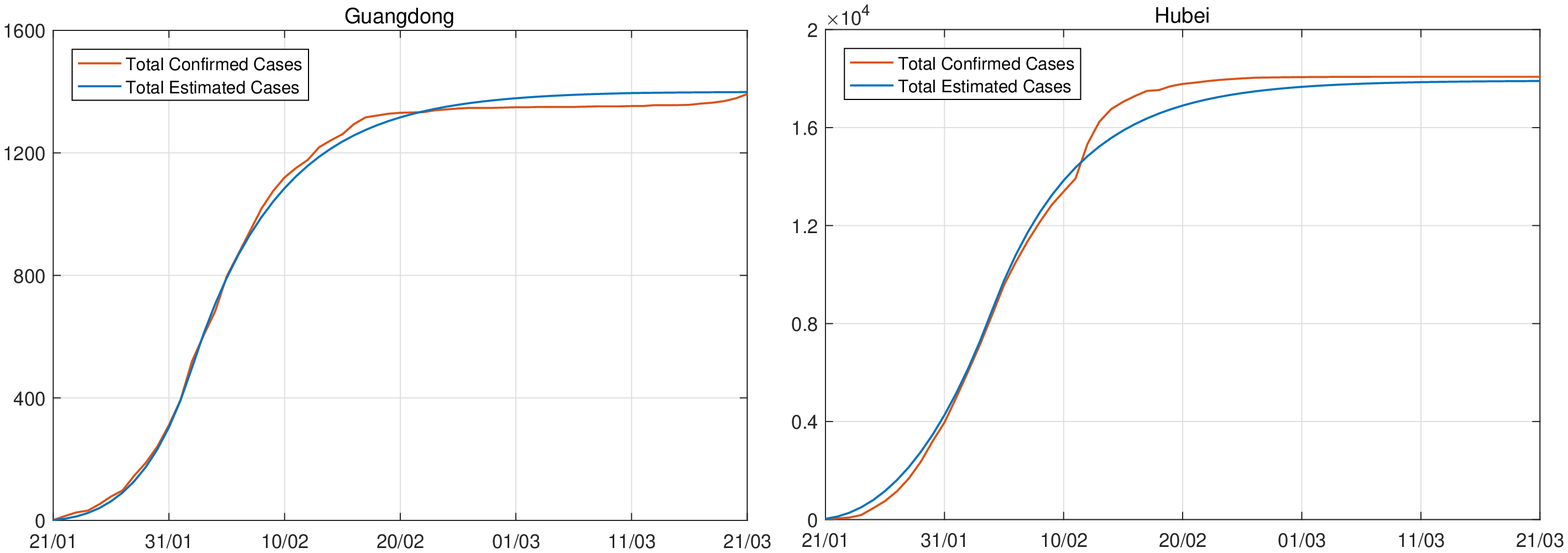}
    \caption{ Comparisons between total confirmed COVID-19 cases and total estimated cases under the calibrated parameters $(\hat{\alpha}, \hat{\beta}, \hat{\delta},\hat{\varrho}, \hat{\ell})$ in Table \ref{T_parameter_estimation_regions_China}  from 2020-01-19 to 2020-03-31 for Guangdong and Hubei, China.}
    \label{Fig_Total_Cases_China}
    \end{center}
\end{figure}

\subsection*{COVID-19 Pandemic for the World}
From mid-February 2020, the COVID-19 started to spread in other countries around the world. At beginning, only a small number of initial cases were reported for some countries in Europe, South/East Asia and North American. However, lately, several large outbreaks were reported in South Korea, Italy, Iran, Spain, Japan and the total number of cases outside China quickly passed the China's. The WHO then recognized the spread of COVID-19 as pandemic on 2020-03-11. We could use this as a second example to confirm our observations from the last exercise. The calibration settings were the same as the previous one. We use the reported daily confirm cases for different regions and countries around world from mid-February to early May  2020. Note that, due to the fact that the pandemic reached each country or territory at different time and the corresponding government interventions also imposed and came into effect at different times, there is no sense to calibrate the model using the data within the same truncated time series. Table \ref{T_parameter_estimation_world} presents the estimation results $(\hat{\alpha},\hat{\beta},\hat{\delta},\hat{\varrho}, \hat{\ell})$ of $(\alpha,\beta,\delta, \varrho, \ell)$ for various countries and territories. We notice that regions and countries like Italy, China, New York have much larger $\varrho$ compared with other areas. This phenomena is reasonable as these areas have specific outbreak area which created external shocks to other part of the regions and countries and hence the number of confirmed cases increased more rapidly than other regions and countries. \\

\begin{table}[hbt!]\footnotesize 
\tabcolsep 0pt  \vspace*{-5pt}
\begin{center}
\def\temptablewidth{1\textwidth}\caption{Calibration parameters $(\hat{\alpha}, \hat{\beta}, \hat{\delta},\hat{\varrho}, \hat{\ell})$  for total confirmed COVID-19 cases from mid-February, 2020 to early May, 2020.    }
{\rule{\temptablewidth}{1pt}}
\begin{tabular*}{\temptablewidth}{@{\extracolsep{\fill}}c|cc|ccccc|c}
\hline
\backslashbox{~~Regions~~}{~~ ~~}& $\textit{Date}_0$ & $\textit{Date}_G$ & $\hat{\alpha}$ & $\hat{\beta}$  & $\hat{\delta}$ & $\hat{\varrho}$   & $\hat{\ell}$ & $\frac{\textsf{MSE}(\hat{\alpha}, \hat{\beta}, \hat{\delta},\hat{\varrho}, \hat{\ell})}{N}$ \\ \hline
  \hline
Australia    & 2020-02-27  & 2020-03-15  & 1.703807 & 3.316083 & 0.401550 & 0.489762  & 28 & 0.124391 \\
\hline
Austria      & 2020-03-01  & 2020-03-10  & 2.860291 & 7.781922 & 0.222730 & 5.788434  & 24 & 0.326694 \\
\hline
China(Mainland)       & 2020-01-19  & 2020-01-23  & 2.846236 & 6.534269 & 0.242773 & 83.160865 & 17 & 1.003152 \\
\hline
Czech        & 2020-03-04  & 2020-03-10  & 3.464654 & 7.842862 & 0.174990 & 2.247588  & 25 & 0.359398 \\
\hline
France       & 2020-02-25  & 2020-03-13  & 3.453459 & 6.912410 & 0.191890 & 15.492535 & 36 & 5.410160 \\
\hline
Germany      & 2020-02-24  & 2020-03-12  & 3.246798 & 6.402097 & 0.201327 & 22.863727 & 32 & 3.464682 \\
\hline
Greece       & 2020-02-26  & 2020-03-10  & 4.181831 & 7.778144 & 0.200886 & 1.337306  & 35 & 0.172176 \\
\hline
Hong Kong           & 2020-03-01  & 2020-03-23  & 2.782369 & 8.054402 & 0.287048 & 0.512120  & 31 & 0.041714 \\
\hline
Iceland      & 2020-02-28  & 2020-03-13  & 3.405901 & 7.908691 & 0.275378 & 1.892734  & 33 & 0.163572 \\
\hline
Italy        & 2020-02-21  & 2020-03-05  & 3.068103 & 5.445646 & 0.211048 & 23.830431 & 30 & 0.911767 \\
\hline
Latvia       & 2020-03-07  & 2020-03-13  & 4.227388 & 7.819275 & 0.208209 & 1.145524  & 22 & 0.099024 \\
\hline
New York     & 2020-02-29  & 2020-03-12  & 2.449906 & 3.203563 & 0.356592 & 70.937604 & 33 & 2.926862 \\
\hline
New Zealand  & 2020-03-12  & 2020-03-16  & 2.122999 & 8.438649 & 0.204667 & 0.551023  & 14 & 0.064829 \\
\hline
Norway       & 2020-02-21  & 2020-03-12  & 4.496501 & 7.519573 & 0.195859 & 6.553447  & 29 & 0.303364 \\
\hline
South Korea  & 2020-02-16  & 2020-02-20  & 1.309507 & 8.771968 & 0.244316 & 5.641225  & 13 & 0.418018 \\
\hline
Switzerland  & 2020-02-25  & 2020-03-13  & 3.195829 & 7.046864 & 0.202615 & 14.083749 & 28 & 1.571950 \\
\hline
\hline
\end{tabular*}\label{T_parameter_estimation_world}
{\rule{\temptablewidth}{1pt}}
\end{center}
\end{table}

In Table \ref{T_parameter_estimation_world}, we also presents the date of day $0$, i.e. $\textit{Date}_0$, and the date of government interventions imposed, i.e. $\textit{Date}_G$. The delay period of government interventions came into effect therefore can be obtained given the estimated $\hat{\ell}$, with $Date_0$ and $Date_G$.  The details for the delay periods of regions and countries are illustrated in Figure  \ref{Fig_Delay_Intervention_world}. We can see that the delay of the interventions for different regions and countries is around $8\sim21$ days. In fact, the delay period can be considered as a criteria for evaluating the effectiveness of restrictions imposed by the authorities to prevent further spread of COVID-19.  In general, most regions and countries with short delay periods normally took tougher restrictions or more effective measures to stop the spread of virus. New Zealand and South Korea are two typical examples. The authority of New Zealand introduced a nationwide lockdown by closing all borders and entry ports to all nonresidents. On contrary, the South Korea authority introduced one of the largest and best-organised epidemic control programs to screen the mass population for the virus with isolation, tracing, quarantine took place simultaneously without further lockdown. Therefore, the estimated delay periods for these two countries are only $9$ days for New Zealand and $8$ days for South Korea, which are much shorter than the average incubation time of COVID-19. For most European countries, due to the containment restriction measures such as quarantines and curfews were not strictly put into effect, the associated delay periods are relatively longer than the incubation time of COVID-19.  \\

The estimated branching ratios before and after the government interventions came into effect for different countries and territories are reported in Table  \ref{T_BR_world}. And Figure \ref{Fig_BR_world} demonstrates a comparison between the BRs before the government interventions took effect and the BRs after the interventions worked, with a blue dash line  that represents $R_b=R_a$. We can immediately see that for most regions and countries, the BR   dropped dramatically after government interventions came into effect, which suggests that the restriction measures imposed by the authority indeed reduce the contagion/infection rate significantly. \\

\begin{figure}[hbt!]
	\begin{center}
		\includegraphics[width=1\textwidth]{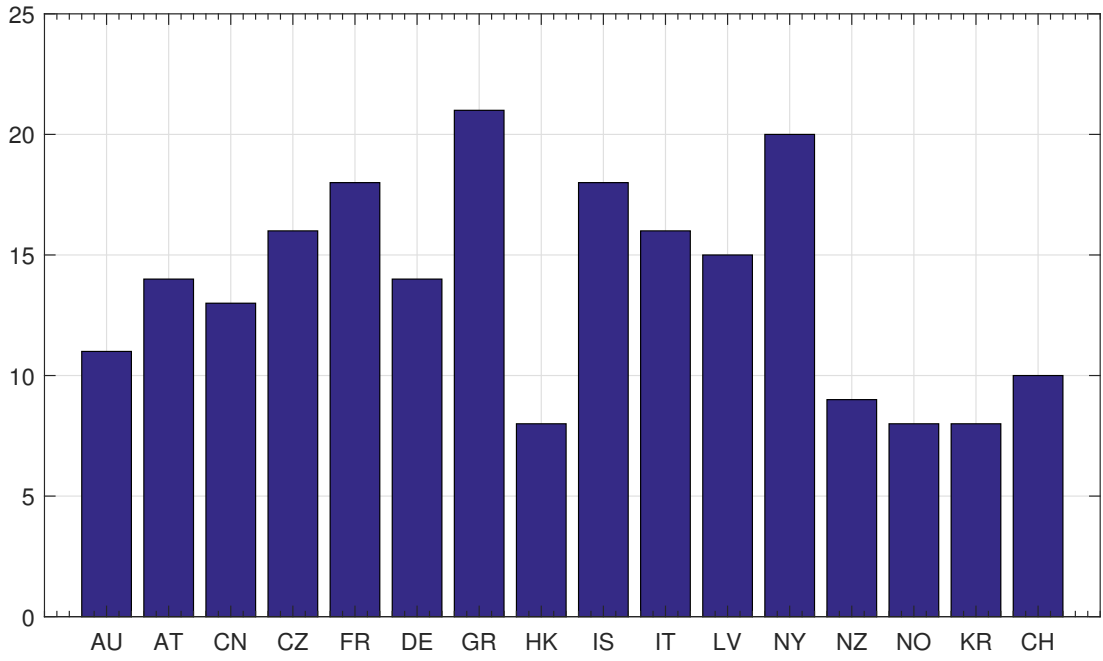}
		\caption{Comparison of the delay period for different regions and countries around the world. The horizontal axis represents  the abbreviation of regions and countries listed in Table \ref{T_BR_world} and the vertical axis represents the number of days for the government interventions came into effect after the government announcement the relevant measures.    }
		\label{Fig_Delay_Intervention_world}
	\end{center}
\end{figure}

\begin{figure}[hbt!]
	\begin{center}
		\includegraphics[width=1\textwidth]{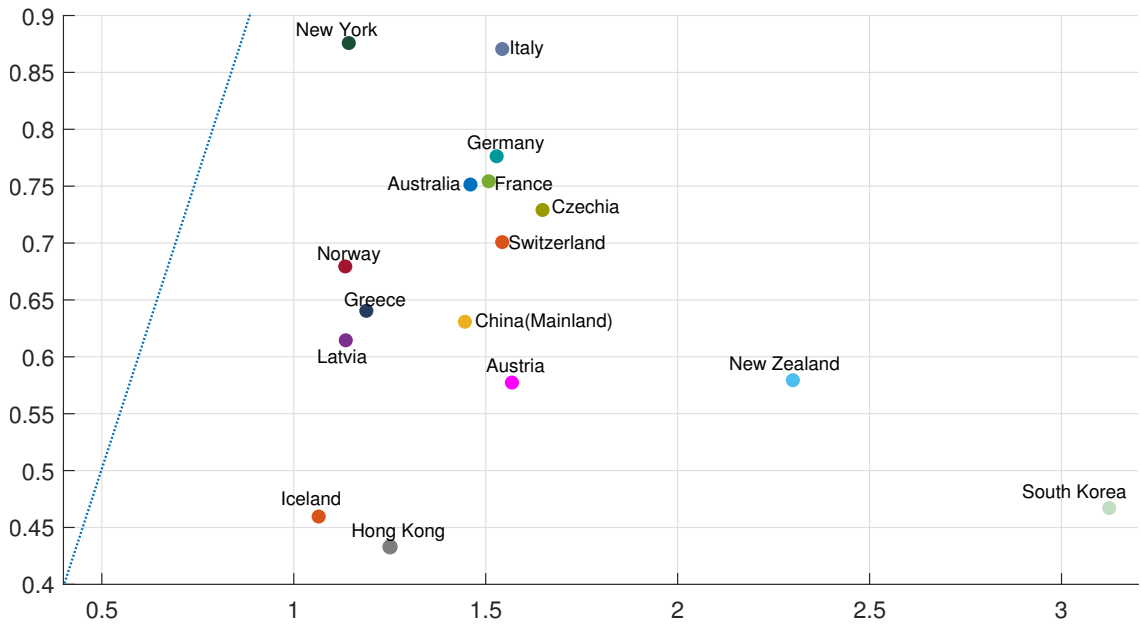}
		\caption{Comparison of the branching ratios {\it before} and {\it after} the government interventions came into effect for different regions and countries around the world. The horizontal axis represents the branching ratio {\it before} the government interventions came into effect, $R_b$, and the vertical axis represents the branching ratio {\it after} the government interventions came into effect, $R_a$. }
		\label{Fig_BR_world}
	\end{center}
\end{figure}

\begin{table}[hbt!]\footnotesize 
\tabcolsep 0pt  \vspace*{-5pt}
\begin{center}
\def\temptablewidth{0.6\textwidth}\caption{The estimated branching ratio (BR) {\it before} and {\it after} the government interventions came into effect, namely $R_b, R_a$ respectively, for regions and countries around the would.   }
{\rule{\temptablewidth}{1pt}}
\begin{tabular*}{\temptablewidth}{@{\extracolsep{\fill}}c|cc}
\hline
\backslashbox{~~Regions~~}{~~BR~~}    & $R_b$      & $R_a$       \\
\hline
 \hline
 Australia            & 1.461638 & 0.750991      \\
 \hline
 Austria              & 1.569681 & 0.576945      \\
 \hline
 China(Mainland)      & 1.447200 & 0.630380    \\
 \hline
 Czechia              & 1.649399 & 0.728637   \\
 \hline
 France               & 1.509016 & 0.753908     \\
 \hline
 Germany              & 1.529826 & 0.775845      \\
 \hline
 Greece               & 1.190377 & 0.639993      \\
  \hline
 Hong Kong            & 1.252077 & 0.432526      \\
  \hline
 Iceland              & 1.066200 & 0.459162    \\
  \hline
 Italy                & 1.544364 & 0.870102    \\
 \hline
 Latvia               & 1.136798 & 0.614084    \\
  \hline
 New York             & 1.144667 & 0.875377  \\
 \hline
 New Zealand          & 2.301455 & 0.579001      \\
  \hline
 Norway               & 1.135488 & 0.678991     \\
 \hline
 South Korea          & 3.125643 & 0.466606      \\
 \hline
 Switzerland         & 1.544349  & 0.700379    \\
\hline
\hline
\end{tabular*}\label{T_BR_world}
{\rule{\temptablewidth}{1pt}}
\end{center}
\end{table}

A comparison between the expected daily confirmed cases for the two-phase dynamic contagion model under the calibrated parameters $(\hat{\alpha}, \hat{\beta}, \hat{\delta},\hat{\varrho}, \hat{\ell})$ is reported in Table \ref{T_parameter_estimation_world}, and the actual daily confirmed COVID-19 cases for different regions and countries over the period of mid-February to early May are presented in Figure \ref{Fig_Comparison_Model_Calibration_world}. In general, we can see that the model can precisely catch the trend of infection, this further confirms that the two-phase dynamic contagion model is effective. Note that, we have smoothed the biggest jump of daily confirmed cases in China for better illustration and fitting purpose. This is due to a change in the confirmation standard established by the Chinese authority on 2020-02-12.\\

\begin{figure}[hbt!]
	\begin{center}
    \includegraphics[width=1\textwidth]{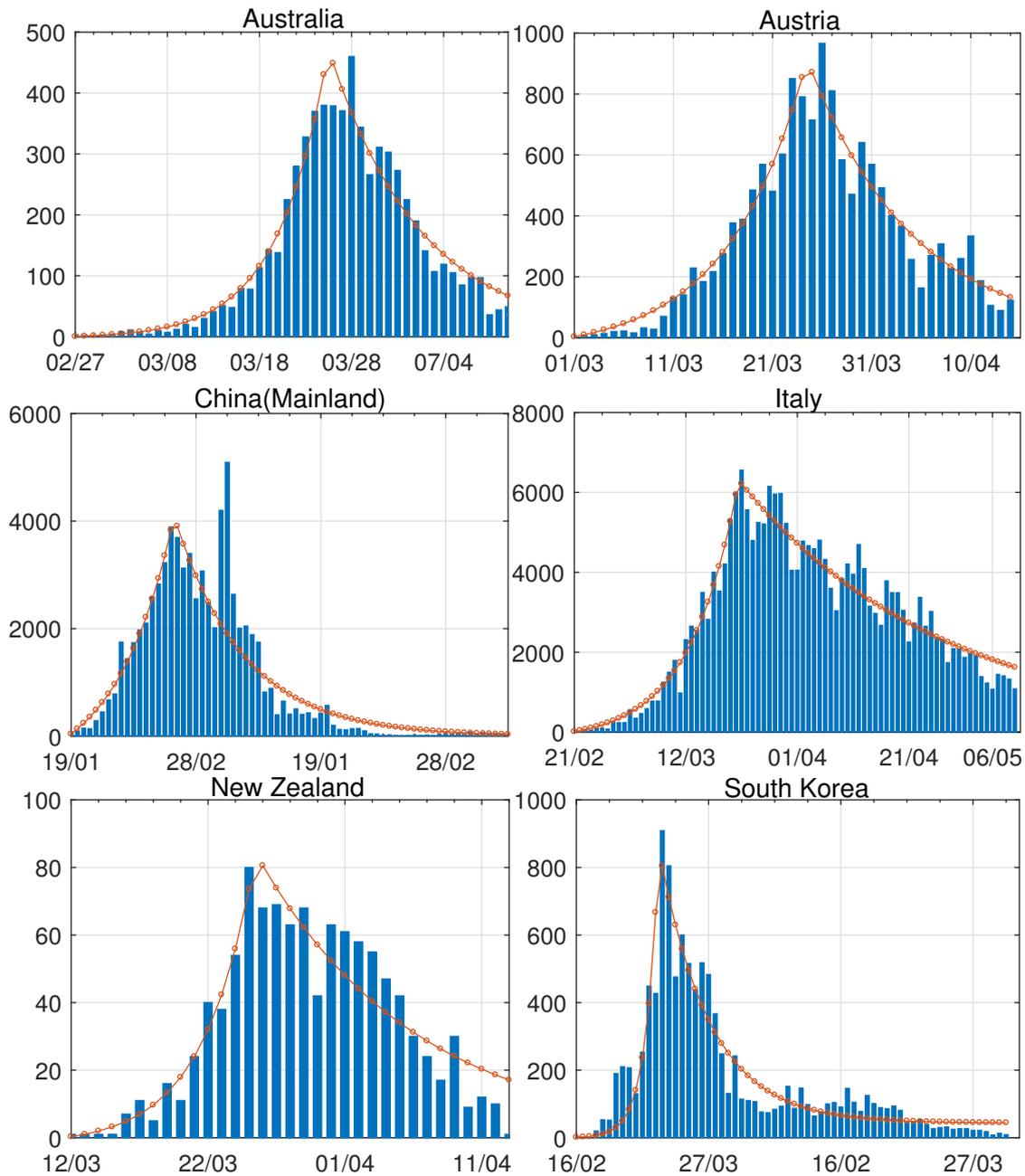}
\caption{Model calibration comparisons between the expected daily confirmed cases under the calibrated parameters  $(\hat{\alpha}, \hat{\beta}, \hat{\delta},\hat{\varrho}, \hat{\ell})$ in Table \ref{T_parameter_estimation_world} and  actual  daily confirmed COVID-19 cases for Australia, Austria, China(Mainland), Italy, New Zealand, and South Korea. }
		\label{Fig_Comparison_Model_Calibration_world}
	\end{center}
\end{figure}

In Figure \ref{Fig_Comparison_Model_Calibration_world_2}, we compare  the estimated daily increment with the actual confirmed COVID-19 cases for France, Germany, Switzerland, and New York. The daily records of confirmed cases for these areas were not as smooth as those countries illustrated in Figure \ref{Fig_Comparison_Model_Calibration_world}. The spikes within the graphs could be caused by many reasons such as the delay of reports, testing capacity, hospital capacity, diagnostic methods and etc. For instance, the daily confirmed cases for France, Germany, Switzerland and New York suddenly declined on a regular basis, which mostly happened during the weekends. Even so, we can see the model can still capture the trend of infectious evolution.
In Figure  \ref{Fig_Total_Cases_World1}, we also compare the actual total confirmed COVID-19 cases against the cumulative estimated cases with a confidence interval within two standard deviations{\footnote{The standard deviations can be derived based on results in Proposition \ref{eq_condtional_expection_lambda_n} and \ref{eq_condtional_expection2_lambda_n}. }} for  some typical countries and territories. The black dash line in each graph of Figure \ref{Fig_Total_Cases_World1} represents the end of data collection period for calibration. The red curve on the left of the black dash line shows the historical data used for calibration and the blue curve is the corresponding estimated result. The red and blue curves on the right of the black dash line demonstrate a comparison between the predicted and actual confirmed COVID-19 cases for countries and regions from the end of their data collection period to the end of May, 2020.  We can see that the estimated curves of the number of confirmed infections under the two-phase dynamic contagion model well fitted the actual propagation process of the COVID-19.   In addition, the forecasted infection cases in the coming weeks after the end of data collection period also well suited the up to date actual total confirmed COVID-19 cases.

\begin{figure}[hbt!]
	\begin{center}
     \includegraphics[width=1\textwidth]{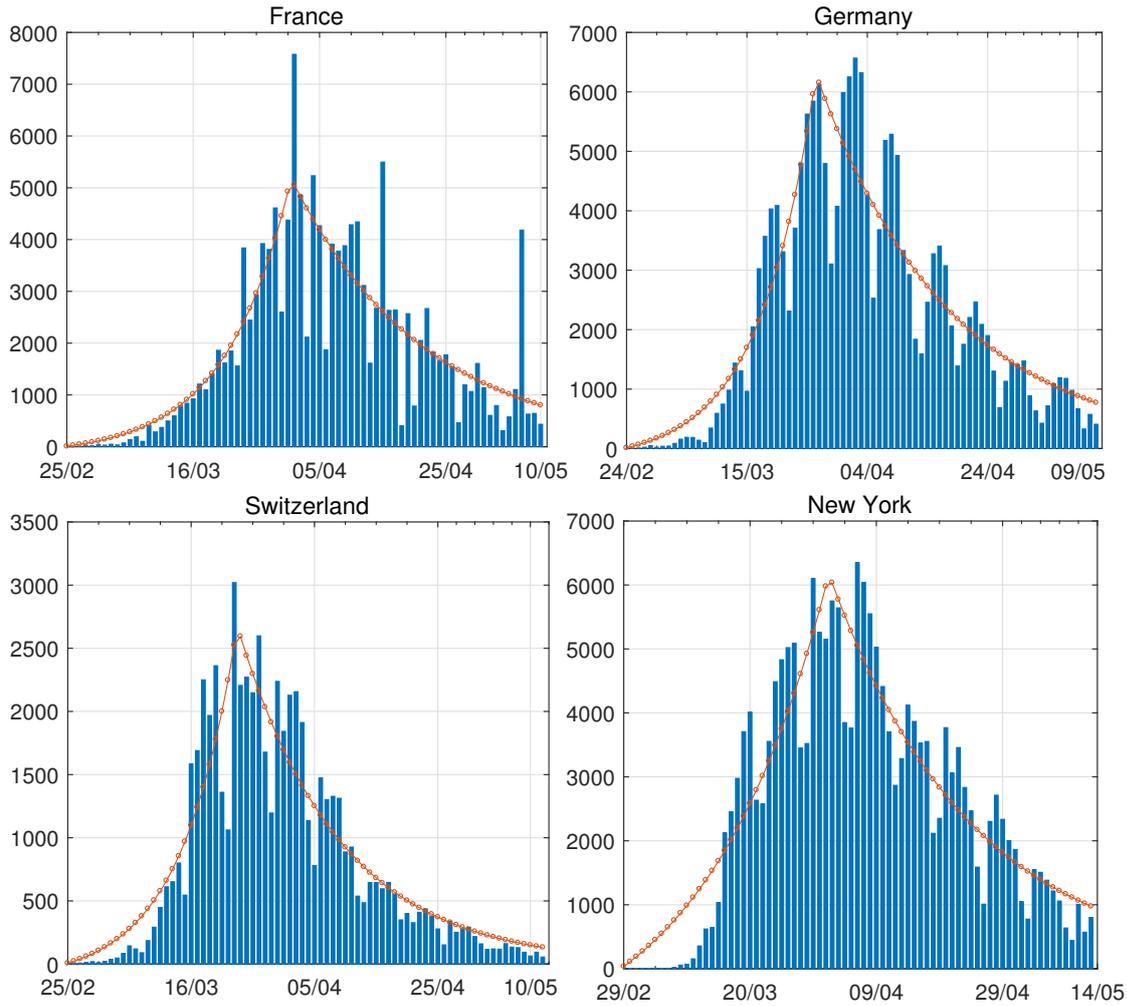}
\caption{Model calibration comparisons between the expected daily confirmed cases under the calibrated parameters  $(\hat{\alpha}, \hat{\beta}, \hat{\delta},\hat{\varrho}, \hat{\ell})$ in Table \ref{T_parameter_estimation_world} and  actual  daily confirmed COVID-19 cases for France, Germany, Switzerland, New York. }
		\label{Fig_Comparison_Model_Calibration_world_2}
	\end{center}
\end{figure}

\begin{figure}[hbt!]
\begin{center}
   \includegraphics[width=1\textwidth]{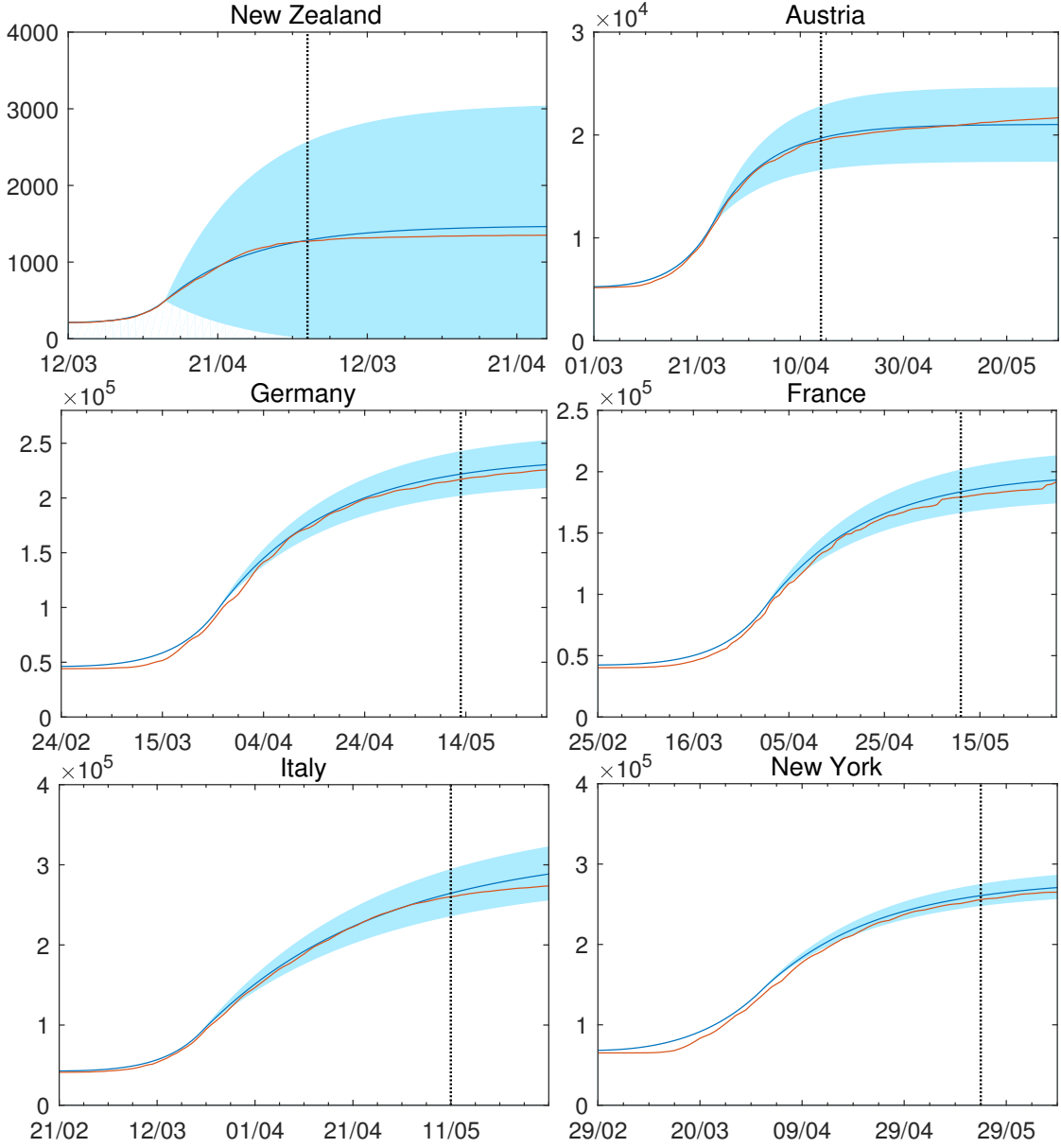}
   \caption{ Comparisons between total confirmed COVID-19 cases and total estimated cases under the calibrated parameters $(\hat{\alpha}, \hat{\beta}, \hat{\delta},\hat{\varrho}, \hat{\ell})$ in Table \ref{T_parameter_estimation_world} from mid-February onward for New Zealand, Austria, Germany, France, Italy, and New York. The red curve represents total confirmed COVID-19 cases. The blue curve represents the total estimated cases and the left zone of the black dash line illustrates historical data used for calibration, and the right zone demonstrates the predicted estimated cases.  The shadowed region plots the values within two standard deviations. }
    \label{Fig_Total_Cases_World1}
    \end{center}
\end{figure}

\subsection{Elimination Probability and Final Epidemic Size}
According to Proposition \ref{prop_distribution_of_elimination_time}, one could obtain the elimination probability of the epidemic by numerically solving the ODE \eqref{eq_ode_lf_pgf}. Based on the calibration parameters  $(\hat{\alpha}, \hat{\beta}, \hat{\delta},\hat{\varrho}, \hat{\ell})$ provided in Table \ref{T_parameter_estimation_world} for regions and countries, we could obtain the associated elimination probabilities. Figure \ref{Fig_Elimination_Prob} illustrates how   $\mathbb{P}\left( \widetilde{T}\le t|\mathcal{F}_{\hat{\ell}}\right)$ varies for different countries and territories. We can see that for regions and countries with effective restriction measures, the probability for a shorter period to observe the last ever event arrives after government interventions come into effect will be much higher. On contrary, for some regions and countries, longer periods are needed for elimination probabilities to be closed to $1$. For instance,  we can see that there is still a long way to go to end the COVID-19 pandemic for Italy.\\

\begin{figure}[hbt!]
	\begin{center}
		\includegraphics[width=0.93\textwidth]{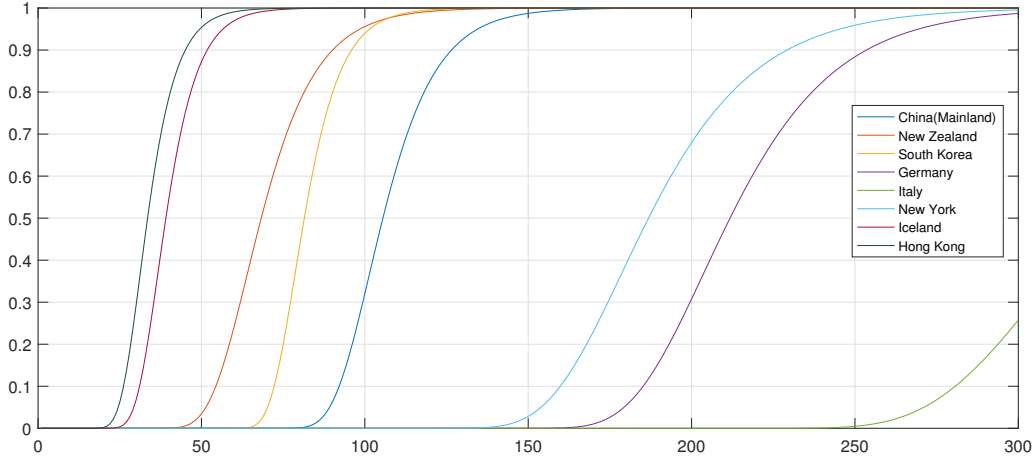}
		\caption{ Conditional elimination probability $\mathbb{P}(\widetilde{T}\le t|\mathcal{F}_{\hat{\ell}})$ for  China (Mainland), New Zealand, South Korea, Germany, Italy, New York, Iceland and Hong Kong under the associated calibration parameters $(\hat{\alpha}, \hat{\beta}, \hat{\delta},\hat{\varrho}, \hat{\ell})$ for these regions and countries suggested in Table \ref{T_parameter_estimation_world}.  }
		\label{Fig_Elimination_Prob}
	\end{center}
\end{figure}

The elimination time of the pandemic depends on many decisive factors, such as the initial intensity of the externally-exciting jumps, the time needed for the government interventions to come into effect, the size of the branching ratio after the government interventions came into effect, etc. Figure \ref{Fig_Time_vs_l}, \ref{Fig_Time_vs_Rho}, and \ref{Fig_Time_vs_BR} illustrate comparisons between the estimated $\hat{\ell}$,  $\hat{\varrho}$, $R_a$ against $\mathbb{E}[\widetilde{T}|\mathcal{F}_{\hat{\ell}}]$, respectively. From Figure \ref{Fig_Time_vs_l}, we can see that for most countries and territories,  the quicker the government interventions come into effect, the faster the pandemic will end.  However, some places like Hong Kong and Iceland still have relatively fast elimination time even though it takes longer for the government interventions to come into effect. This is probably because the  restriction measures for these places were imposed so early that reporting procedures were not properly in place yet. Figure \ref{Fig_Time_vs_Rho}, and \ref{Fig_Time_vs_BR} clearly demonstrate that the externally-exciting jump intensity $\varrho$ and the branching ratio after the government came into effect $R_a$ are important factors that determine the extinction time of the pandemic. In general, to reduce the extinction time of the pandemic, the first priority for the authorities should be introducing restriction measures such as national/subnational lockdown to reduce the intensity of the external imported cases, and while the external imported cases are controlled and thereafter negligible, the governments should simultaneously introduce enforced restrictions to prevent further transmission. If the government intervention strategies were effectively implemented without a lack of civic spirit, the infection rate of the virus after the these intervention measures come into place will be reduced significantly, and therefore lead to a quicker elimination of the COVID-19 pandemic. Note that the prediction for expected elimination time for regions and countries is based on the assumption that the government intervention measures are still taking into place in some form and propagation of the disease continues as in  Phase 2. Relaxation of the government intervention measures  will inevitably delay the disease elimination for most regions and countries.
\\

\begin{figure}[hbt!]
	\begin{center}
		\includegraphics[width=1\textwidth]{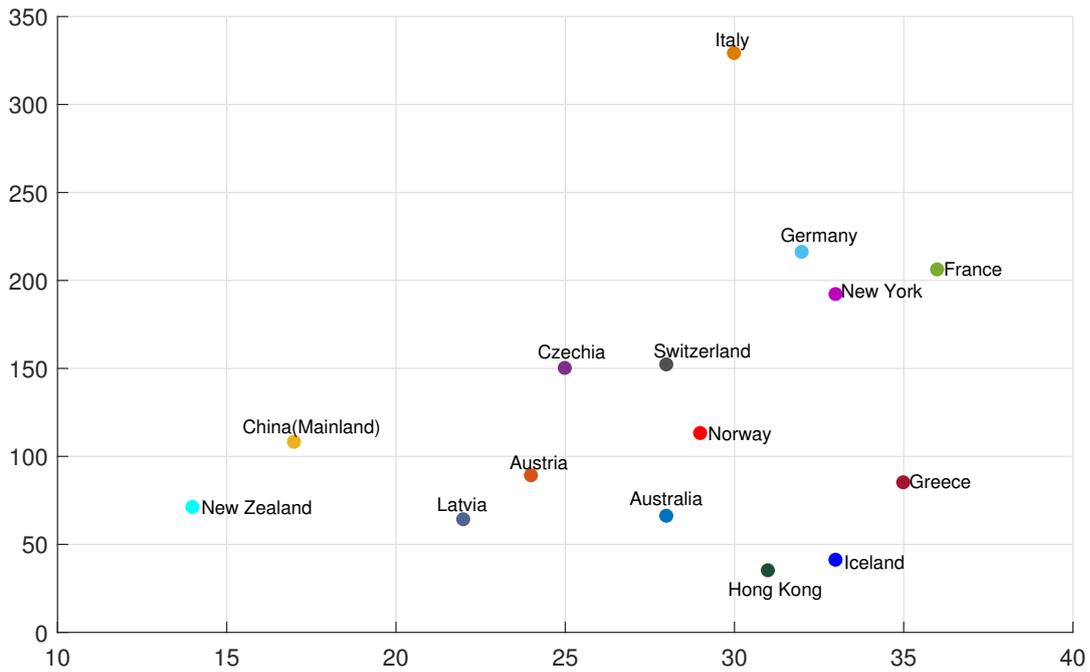}
		\caption{Comparison of the conditional expectation of the elimination time $\widetilde{T}$ and the estimated government interventions came into effect time for different regions and countries around the world. The horizontal axis represents $\hat{\ell}$, and the vertical axis represents $\mathbb{E}[\widetilde{T}|\mathcal{F}_{\hat{\ell}}]$  .}
		\label{Fig_Time_vs_l}
	\end{center}
\end{figure}

\begin{figure}[hbt!]
	\begin{center}
		\includegraphics[width=1\textwidth]{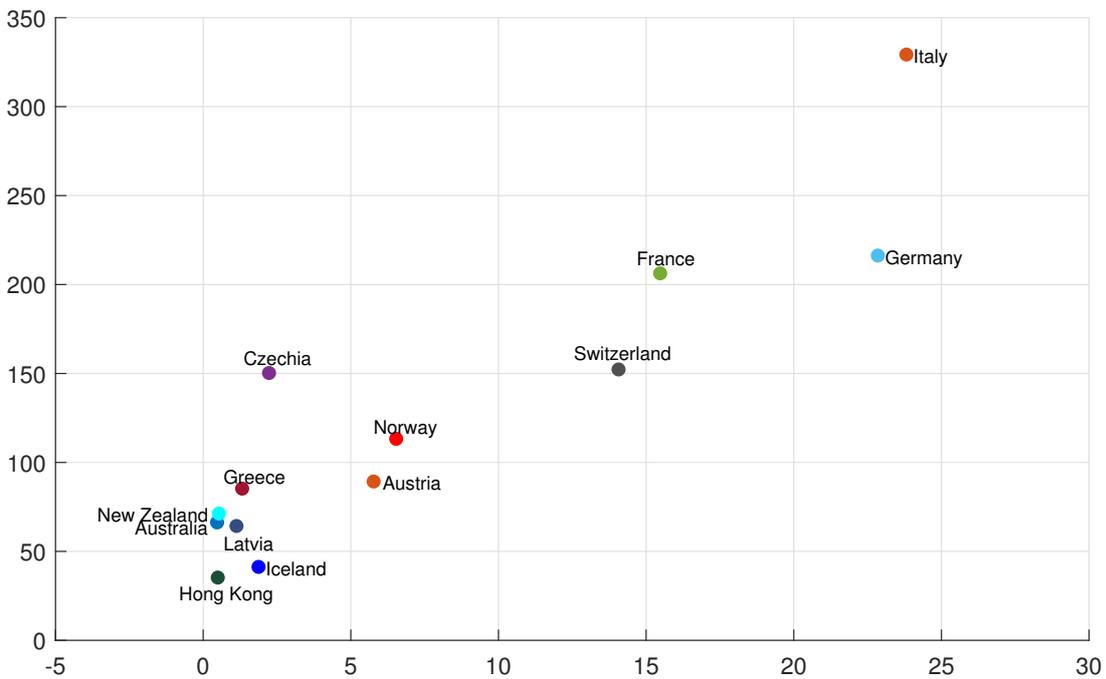}
		\caption{Comparison of  the conditional expectation of the elimination time $\widetilde{T}$ and the estimated intensity for externally-exciting jumps for different regions and countries around the world. The horizontal axis represents $\hat{\varrho}$, and the vertical axis represents $\mathbb{E}[\widetilde{T}|\mathcal{F}_{\hat{\ell}}]$  .}
		\label{Fig_Time_vs_Rho}
	\end{center}
\end{figure}

\begin{figure}[hbt!]
	\begin{center}
		\includegraphics[width=1\textwidth]{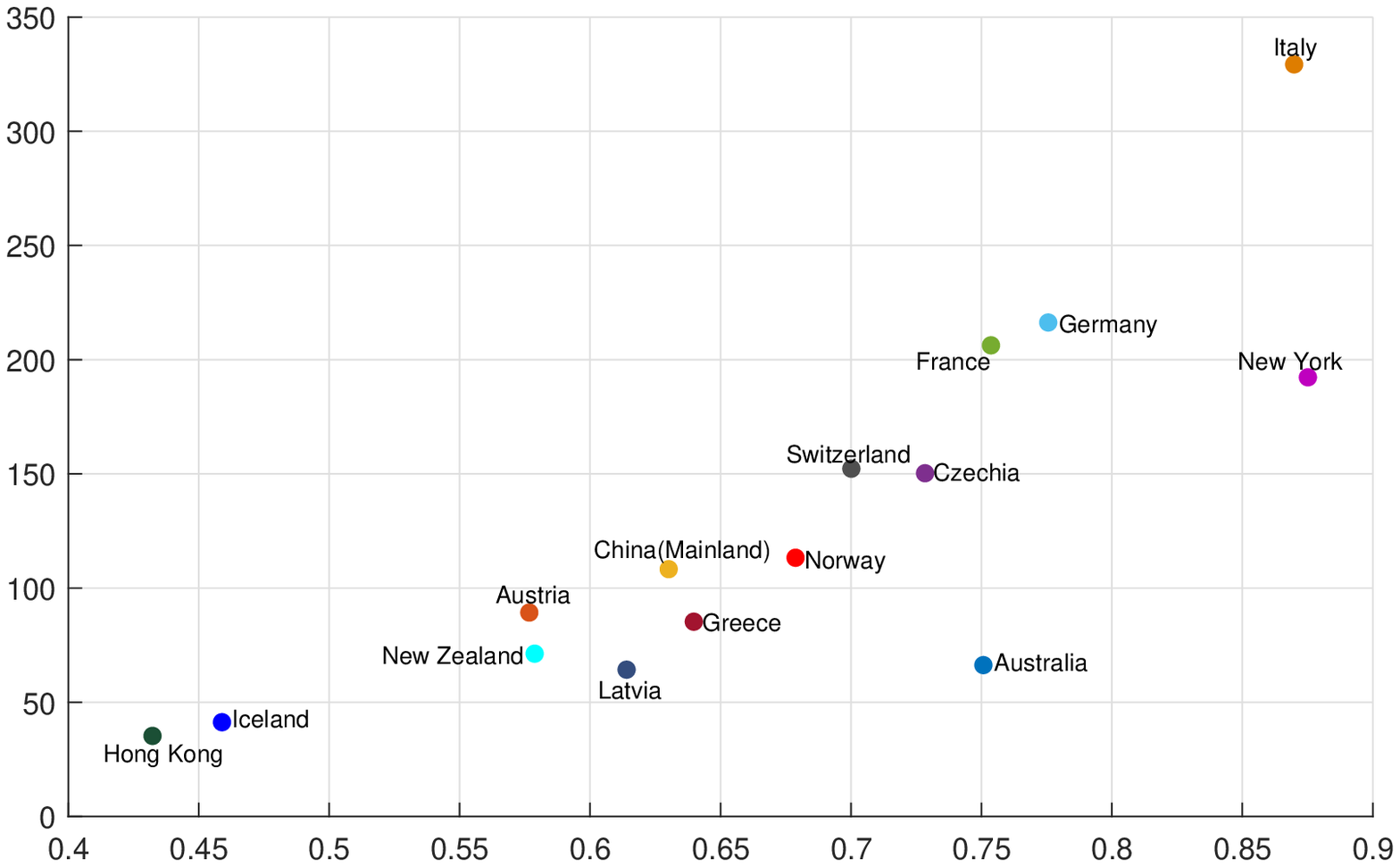}
		\caption{Comparison of the conditional expectation of the elimination time $\widetilde{T}$ and the branching ratio {\it after} the government interventions came into effect for different regions and countries around the world. The horizontal axis represents  $R_a$, and the vertical axis represents $\mathbb{E}[\widetilde{T}|\mathcal{F}_{\hat{\ell}}]$. }
		\label{Fig_Time_vs_BR}
	\end{center}
\end{figure}

Beside the conditional probability for the elimination time $\widetilde{T}$, the epidemic size $N_t$  given $\{\tilde{T}\le t\}$ can also be predicted according to the join expectation of $N_t$ and $\{\widetilde{T}\le t\}$ derived in Corollary \ref{prop_expectation_epidemic_size}. In Table \ref{T_Elimination_Prob_and_Size}, we report the $95\%$ confidence interval for elimination time $\widetilde{T}$  the condition expectation of the elimination time $\widetilde{T}$,  $\mathbb{E}[\widetilde{T}|\mathcal{F}_{\hat{\ell}}]$, the expected elimination date $\textit{Date}_E$, and the conditional expectation of the epidemic size $N_t$,   $\mathbb{E}[N_t | \mathcal{F}_{\hat{\ell}}\cap \{\widetilde{T}\le t\} ]$ with $t=\mathbb{E}[\widetilde{T}|\mathcal{F}_{\hat{\ell}}]$, for regions and countries with calibration parameters in Table \ref{T_parameter_estimation_world}. Note that, the regions and countries with more confirmed COVID-19 cases before government interventions came into effect will experience longer time to reach elimination state, like France, Germany, Italy, and New York. And the corresponding expected epidemic size for these areas are also much larger. Note that since we have smoothed the biggest jump of daily confirmed cases, adding up with the cases which have been smoothed,  the actual conditional expectation of the epidemic size is about $83113$, which is very close to the current total confirmed cases $82,993$ on 2020-05-27. In general, not only the expected epidemic size is very close to the actual total confirmed cases for the listed regions and countries, but also the estimated elimination date is very close to the actual eradicate date.  New Zealand is one typical example that can be used to demonstrate the effectiveness of our model in predicting the key epidemiological quantities. New Zealand authority has officially declared that the country has completely eradicated COVID-19 for now with total of $1154$ confirmed COVID-19 cases on 2020-06-08. This elimination date and the final epidemic size are very close to what we predicted for New Zealand under the two-phase dynamic contagion model, the predicted elimination date is around 2020-06-04 and the predicted epidemic size is about $1250$.  More remarkably, the historical data  we used for model calibration for New Zealand is from 2020-03-12 to 2020-04-13, which is already one and half months ago. This clearly shows that the two-phase dynamic contagion model is pretty powerful in forecasting cumulative confirmed COVID-19 cases, predicting possible elimination duration for the pandemic, and evaluating effectiveness of relevant government intervention measures.\\


\begin{table}[hbt!]\footnotesize 
\tabcolsep 0pt  \vspace*{-5pt}
\begin{center}
\def\temptablewidth{1\textwidth}\caption{The $95\%$ confidence interval for elimination time $\widetilde{T}$, the conditional expectation of the elimination time $\widetilde{T}$, the expected elimination date and the conditional expectation of the epidemic size $N_t$ given $\widetilde{T}\le t$ with $t=\mathbb{E}[\widetilde{T}|\mathcal{F}_{\hat{\ell}}]$ under the calibrated parameters  $(\hat{\alpha}, \hat{\beta}, \hat{\delta},\hat{\varrho}, \hat{\ell})$ for these regions and countries suggested in Table \ref{T_parameter_estimation_world}.  }
{\rule{\temptablewidth}{1pt}}
\begin{tabular*}{\temptablewidth}{@{\extracolsep{\fill}} c|c|c|c|c  }
\hline
\backslashbox{~~Regions~~}{~ ~} & $\mathbb{P}( \widetilde{T}\in(t_1, t_2)|\mathcal{F}_{\hat{\ell}})=95\%$ & $\mathbb{E}[\widetilde{T}|\mathcal{F}_{\hat{\ell}}]$  ~~~~   & $\textit{Date}_E$  & $\mathbb{E}[N_t | \mathcal{F}_{\hat{\ell}}\cap \{\widetilde{T}\le t\} ]$      \\ \hline
\hline
Australia        & $(47,97)$   &  66      &  2020-05-30   & 7059.46                         \\
\hline
Austria          & $(69,122)$  &  89      &  2020-06-21   & 15610.22                        \\
\hline
China (Mainland) & $(87,142)$  &108      &  2020-05-22   & 69675.10                         \\
 \hline
Czechia          & $(111,216)$ &150      &  2020-08-25   & 8932.448                      \\
\hline
France           & $(166,272)$ &206      &    2020-10-23 & 154112.00                       \\
\hline
Germany          & $(175,285)$ &216      &   2020-10-28  & 189116.50                      \\
\hline
Greece           & $(59,127)$  &85       &  2020-06-24   & 2650.07                     \\
\hline
Hong Kong        & $(23,54)$   &35       &  2020-05-05   & 945.81                         \\
\hline
Iceland          & $(28,61)$   &41       &  2020-05-11   & 1787.79                     \\
\hline
Italy            &$(264,445)$  &329      &  2021-02-13   & 276480.1                       \\
\hline
Latvia           &$(40,102)$   &64       &  2020-05-31   & 764.80                        \\
\hline
New York         &$(149,261)$  &192       &  2020-10-10   & 209311.80                     \\
\hline
New Zealand      &$(49,107)$   &71        &  2020-06-04   & 1249.77                        \\
\hline
Norway           &$(83,162)$   &113       &   2020-07-15  & 8024.73                       \\
\hline
Switzerland     & $(121,203)$    &152       &  2020-08-22   & 33414.19                    \\
\hline
\end{tabular*}\label{T_Elimination_Prob_and_Size}
{\rule{\temptablewidth}{1pt}}
\end{center}
\end{table}

The conditional distribution of the final epidemic size $N_\infty$ can be obtained by numerically inverse the probability generating function provided in Proposition \ref{prop_final_epidemic_size_distribution}. Since we assume the self-exciting jumps follows an exponential distribution after government interventions came into effect, the final epidemic size $N_\infty$ conditional on $\mathcal{F}_{\ell}$ follows a mixed-Poisson distribution with probability mass function specified in \eqref{eq_mix_poisson_distribution}. Figure \ref{Fig_Final_Epidemic_Size} demonstrates the conditional probability mass function of the difference between the finial epidemic size $N_\infty$ and the total number of confirm cases $N_{\ell}$ when government interventions came into effect, i.e., $\mathbb{P}\left(N_\infty-N_{{\ell}} =k\mid \mathcal{F}_{{\ell}} \right)$, for some regions and countries under the calibrated parameters  $(\hat{\alpha}, \hat{\beta}, \hat{\delta},\hat{\varrho}, \hat{\ell})$ for these regions and countries suggested in Table \ref{T_parameter_estimation_world}.

\begin{figure}[hbt!]
	\begin{center}
		\includegraphics[width=1\textwidth]{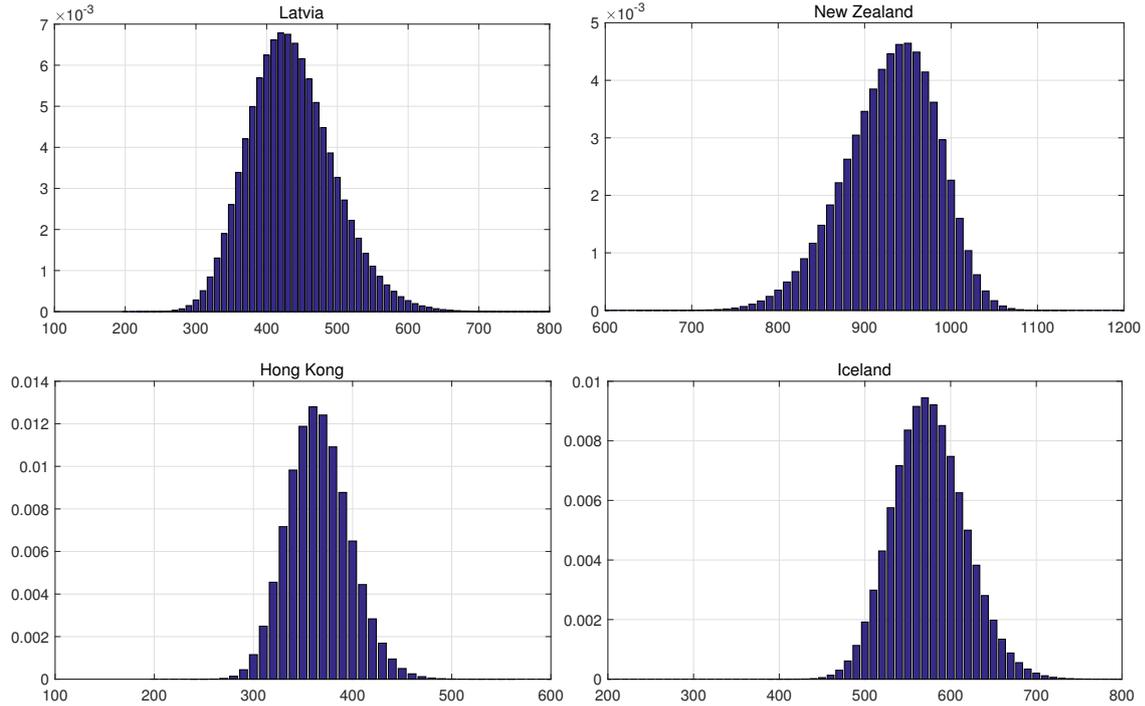}
		\caption{ Probability mass function  $\mathbb{P}\left(N_\infty-N_{\hat{\ell}} =k\mid \mathcal{F}_{\hat{\ell}} \right)$  for Latvia, New Zealand, Hong Kong and Iceland   under the calibrated parameters  $(\hat{\alpha}, \hat{\beta}, \hat{\delta},\hat{\varrho}, \hat{\ell})$ for these regions suggested in Table \ref{T_parameter_estimation_world}.}
		\label{Fig_Final_Epidemic_Size}
	\end{center}
\end{figure}

\section{Concluding Remarks}\label{Sec_Conclusion}

In this paper, we have introduced a two-phase dynamic contagion process for modelling the current spread of COVID-19.
This model allows randomness to the infectivity of individuals rather than a constant reproduction number as commonly assumed by standard models.
Key episdemiological quantities, such as the distribution of final epidemic size and  expected epidemic duration, are derived and estimated based on real data for various regions and countries. In addition, the associated time lag of the effect of intervention in each country or  region has been estimated, and our empirical results are consistent to the incubation time of COVID-19 for most people found by existing medical study such as \citet{lauer2020incubation}.
The aim of this paper is to demonstrate that our model, as a representative of Hawkes-based processes, could  be a valuable tool for epidemic modelling. In fact, the vast literature of Hawkes-based processes would also be relevant and potentially be applicable.
For example, multivariate extensions of Hawkes-based processes, such as \citet{embrechts2011multivariate} for analysing financial high-frequency data, could be adopted for  modelling the cross-region epidemic contagion. L{\'e}vy-driven extensions,  such as \citet{qu2019efficient} for portfolio credit risk, may perform better in capturing the heavy tail property of epidemic distribution.  In addition, easing of the government interventions will lead to change of parameters and delay extinction times. The model can be adjusted by introducing an additional phase with change of parameters.  When countries cycle between periods of restrictions and relaxations to manage COVID-19, we can adjust the two-phase dynamic contagion model by replacing the constant parameters to  piecewise time dependent parameters. These are all proposed as future research.

 
\small
\renewcommand\bibname{References}
\bibliography{BibTex_Hongbiao}
\bibliographystyle{apa} 

\end{document}